\documentclass[a4paper, 9pt]{llncs}

\usepackage{amssymb}
\usepackage{graphicx}
\usepackage{xspace}
\usepackage{fancyhdr}

\usepackage{url}
\urldef{\mailsa}\path|{alfred.hofmann, ursula.barth, ingrid.haas, frank.holzwarth,|
\urldef{\mailsb}\path|anna.kramer, leonie.kunz, christine.reiss, nicole.sator,|
\urldef{\mailsc}\path|erika.siebert-cole, peter.strasser, lncs}@springer.com|    
\newcommand{\keywords}[1]{\par\addvspace\baselineskip
\noindent\keywordname\enspace\ignorespaces#1}

\usepackage{microtype}
\usepackage{graphicx}
\usepackage{layout}
\usepackage{lipsum}
\usepackage[utf8]{inputenc}
\usepackage{fancyhdr}
\usepackage{xspace}
\usepackage{tabularx}
\usepackage{algorithm}
\usepackage{algorithmicx}
\usepackage{algcompatible}
\usepackage{algpseudocode}
\algrenewcommand\algorithmicindent{0.45em}%
\usepackage[toc,page]{appendix}
\usepackage{color}
\usepackage{caption}
\usepackage{ragged2e}
\usepackage{tabularx}
\usepackage{amsmath,amssymb,latexsym} 
\usepackage{fancyvrb}
\usepackage{float}
\usepackage{cite}
\usepackage{epstopdf}
\usepackage{epsfig}
\usepackage{graphicx}
\usepackage{pgfplots}
\usepackage{lipsum}
\usepackage{booktabs,chemformula}
\usepackage{tabularx,booktabs} 
\usepackage{etoolbox}
\usepackage{breqn}
\usepackage{subfiles}
\usepackage{sidecap}
\usepackage{authblk}
\usepackage{soul}
\usepackage{mathtools}

\usepackage{pdfpages}
\pgfplotsset{compat=1.7}
\usepackage{calligra}
\usepackage{calrsfs}
\DeclareMathAlphabet{\mathcalligra}{T1}{calligra}{m}{n}
\DeclareCaptionFont{tiny}{\tiny}
\algrenewcommand\alglinenumber[1]{\tiny #1:}
\usepackage{multicol}
\usepackage{comment}
\usepackage{subcaption}
\usepackage{authblk}
\usepackage{appendix}
\usepackage{xspace}
\usepackage{enumitem}
\usepackage{soul}
\usepackage{mathtools}
\usepackage{breqn}
\usepackage{hyperref}
\usepackage{url}
\hypersetup{
    colorlinks=true,
    linkcolor=blue,
    filecolor=magenta,      
    urlcolor=cyan,
}
 \usepackage{lastpage}
\newcommand{\punt}[1]{}
\newcommand{\cmnt}[1]{}

\newcommand{\lble} {linearizable\xspace}
\newcommand{\lbty} {linearizability\xspace}

\newcounter{history}

\newenvironment{proofsketch}[1][Proof Sketch]{\noindent#1: }{\hfill $\Box$\\[0.4mm]} 

\newcommand{\secref}[1]{Section~\ref{sec:#1}}

\newcommand{\lineref}[1]{Line~\ref{lin:#1}}

\newcommand{\ignore}[1]{}

\algdef{SE}[DOWHILE]{Do}{doWhile}{\algorithmicdo}[1]{\algorithmicwhile\ #1}
\newcommand{\mth} {method\xspace}

\newcommand{\lp} {LP\xspace}

\newcommand{\enode}{enode\xspace}
\newcommand{\vnode}{vnode\xspace}
\newcommand{\vh}{VertexHead}

\newcommand{\addv}{AddVertex\xspace}
\newcommand{\remv}{RemoveVertex\xspace}
\newcommand{\adde}{AddEdge\xspace}
\newcommand{\reme}{RemoveEdge\xspace}

\newcommand{\cone}{ContainsEdge\xspace}

\algrenewcommand{\algorithmiccomment}[1]{$//$ #1}

\newcommand{\ds}{data-structure\xspace}

\newcommand{\valc}{ValidateC\xspace}
\newcommand{\valv}{ValidateV\xspace}
\newcommand{\vale}{ValidateE\xspace}
\newcommand{\val}{Validate\xspace}

\newcommand{\cchead}{CCHead\xspace}
\newcommand{\cctail}{CCTail\xspace}

\newcommand{\ehead}{EHead\xspace}

\newcommand{\loccc}{locateSCC\xspace}
\newcommand{\addcc}{addSCC\xspace}
\newcommand{\remcc}{removeSCC\xspace}
\newcommand{\addvcc}{addVertexSCC\xspace}
\newcommand{\createcc}{createNewSCCwithNewV\xspace}
\newcommand{\createoldvcc}{createNewSCCwithOldV\xspace}
\newcommand{\addvncc}{addOldVInSCC\xspace}
\newcommand{\mergescc}{mergeSCC\xspace}
\newcommand{\createe}{createE\xspace}
\newcommand{\createc}{createSCC\xspace}
\newcommand{\addenode}{addENode\xspace}
\newcommand{\remenode}{remENode\xspace}
\newcommand{\remvnode}{remVNode\xspace}
\newcommand{\loce}{locateE\xspace}
\newcommand{\locv}{locateV\xspace}

\newcommand{\dfsf}{DFSFW\xspace}
\newcommand{\dfsb}{DFSBW\xspace}
\newcommand{\findsccreme}{findSCCafterRemoveE}
\newcommand{\findsccadde}{findSCCafterAddE}

\newcommand{\checkscc}{checkSCC}
\newcommand{\blongsto}{blongsToCommunity}

\begin{document}

\mainmatter  

\title{Maintenance of Strongly Connected Component in Shared-memory Graph}


%
%
\ignore{
\author{Alfred Hofmann%
\thanks{Please note that the LNCS Editorial assumes that all authors have used
the western naming convention, with given names preceding surnames. This determines
the structure of the names in the running heads and the author index.}%
\and Ursula Barth\and Ingrid Haas\and Frank Holzwarth\and\\
Anna Kramer\and Leonie Kunz\and Christine Rei\ss\and\\
Nicole Sator\and Erika Siebert-Cole\and Peter Stra\ss er}
}
\vspace{-3mm}
\author{
       Muktikanta Sa 
}

\vspace{-3mm}
\institute{ Department of Computer Science \& Engineering \\
      Indian Institute of Technology Hyderabad, India \\
      cs15resch11012@iith.ac.in}
%
%

\toctitle{Lecture Notes in Computer Science}
\tocauthor{Authors' Instructions}
\maketitle
\vspace{-0.7cm}
\begin{abstract}
In this paper, we present an on-line fully dynamic algorithm for maintaining strongly connected component of a directed graph in a shared memory architecture. The edges and vertices are added or deleted concurrently by fixed number of threads. To the best of our knowledge, this is the first work to propose using linearizable concurrent directed graph and is build using both ordered and unordered list-based set. We provide an empirical comparison against sequential and  coarse-grained. The results show our algorithm's throughput is increased between $3$ to $6$x depending on different workload distributions and applications. We believe that there are huge applications in the on-line graph. Finally, we show how the algorithm can be extended to community detection in on-line graph. 


\end{abstract}

\vspace{-10mm}
\section{Introduction}
\label{sec:intro}
\vspace{-0.3cm}
Generally the real-world practical graph always dynamically change over time. Dynamic graphs are the one's which are subjected to a sequence of changes like insertion, deletion of vertices and/or edges \cite{Demetrescu2010}. Dynamic graph algorithms are used extensively and it has been studied for several decades. Many important results have been achieved for fundamental dynamic graph problems and some of these problems are very challenging i.e,finding cycles, graph coloring, minimum spanning tree, shortest path between a pair of vertices, connectivity, 2-edge \& 2-vertex connectivity, transitive closure, strongly connected components, flow network, etc (see, e.g., the survey in \cite{Demetrescu2010}). The biological graph networks are very complicated due to their complex layered architecture and size. Graphs Networks, typically involving finding the perfect match or similarities on gene expression for evolution like disease progression \cite{Joseph12} are even more challenging to the research community.

We have been specifically motivated by largely used  problem of fully dynamic evolution \emph{Strongly Connected Components}(SCC). Detection of SCC in dynamically changing graph affects a large community both in the theoretical computer science and the network community. SCC detection on static networks fails to capture the natural phenomena and important dynamics. Discovering SCCs on dynamic graph helps uncover the laws in processes of graph evolution, which have been proven necessary to capture essential structural information in on-line social networking platforms (facebook, linkedin, google+, twitter, quora, etc.). SCC often merges or splits because of the changing friendship over time. A common application of SCC on these social graph is to check weather two members belong to the same SCC (or community). So, we define the \emph{SCC(id1,id2)}: which checks if there is a directed path from $id1$ to $id1$ and the other way round. In general, a social network graph handles the concurrency control over a set of users or threads running concurrently. A thread as a block of code is invoked by the help of methods to access multiple shared memory objects atomically.

 In this paper, we present a new shared-memory algorithm called as \emph{SMSCC} for maintaining SCC in fully dynamic directed graphs.The following are the key contributions of this work:
\begin{enumerate}
    \item  Firstly, we designed a new incremental algorithm(SMDSCC) for maintaining SCC dynamically. i.e, after inserting an edge or a vertex how quickly we update the SCC.
    \item Secondly, we designed a decremental algorithm(SMISCC) for maintaining SCC dynamically. i.e, after deleting an edge or a vertex how quickly we update the SCC.
    \item  Thirdly, a algorithm for maintaining fully dynamic SCC(SMSCC).
    \item An empirical comparison against sequential and coarse-grained.
    \item Our algorithm is work-efficient for most on-line graphs.
    \item An application suite : community detection on-line graph.
     
\end{enumerate}
\noindent
We have not found any comparable concurrent data-structure for solving this strongly connected components problem in shared-memory architecture. Hence we crosscheck against sequential and coarse-grained implementations. 
\subsection{Background and Related Work}

\noindent Let a concurrent directed graph $G = (V,E)$, $G(V)$ is a set of vertices and set of $G(E)$ directed edges. We use \emph{Adj(u)} to denote the set of neighbors of a vertex $u$. The $G(E)$ is collection of both outgoing and incoming neighbors, i.e., \emph{Adj(u)} = \{ for outgoing edges $v : \langle u, v \rangle$ \& for incoming edges $w: \langle w, u \rangle \in G(E)$\}. Each edge connects an ordered pair of vertices currently belongs to $G(V)$. And this $G$ is dynamically being modified by a fixed set of concurrent running threads. Our dynamic graph setting, threads can perform insertion/deletion of edges and insertion of vertices. We assume that all the vertices have unique identification key, which is captured by $val$ field in Gnode structure as shown in the Section \ref{sec:const-scc-graph}. 
\noindent
\begin{definition}
\hspace{-2mm}\textnormal{(Reachability)},Given a graph G=(V,E) and two vertices $u,v\in G(V)$, a vertex $v$ is reachable from another vertex $u$ if there is a path from $u$ to $v$. 
\end{definition}
\vspace{-3mm}
\begin{definition}
\hspace{-2mm}\textnormal{(SCC)}, A graph is strongly connected if there is a \emph{directed path} from any vertex to every other vertex. Formally,\\
Let $G=(V, E)$ be a directed graph. We say two nodes $u,v\in G(V)$ are called  strongly  connected  iff  $v$ is reachable from $u$ and also $u$ is reachable from $v$. A strongly connected component(or SCC) of $G$ is a set $C \subseteq G(V)$ such that: 
 \begin{enumerate}
 \item $C$  is not empty.
 \item For any  $u, v \in C$: $u$ and $v$ are strongly connected.
 \item For any $u \in C$ and $v \in V-C$: $u$ and $v$ are not strongly connected.
 \end{enumerate}
\end{definition}

 Apart from the definition, SCC also satisfies the equivalence relation on the set of vertices:
\textit{Reflexive:} every vertex $v$ is strongly connected to itself. \textit{Symmetric:} if $u$ is strongly connected to $v$, then $v$ is strongly connected to $u$. \textit{Transitive:} if $u$ is strongly connected to $v$ and $v$ is strongly connected to $w$, then $u$ is also strongly connected to $w$. 
\vspace{-0.4cm}
\subsection{Related Work}
There have been many parallel computing algorithms proposed for computing SCC both in directed and undirected graphs. Hopcroft and Tarjan \cite{HopcroftTarjan:1973} presented the first algorithm to compute the connected components of a graph using the depth first searches(DFS) approaches. Hirschburg et al. \cite{Hirschberg:1979} presented a  novel parallel algorithm for finding the connected components in an undirected graph. Shiloach and Vishkin \cite{ShiloachVishkin:1982} proposed an parallel computing \emph{O(logV)} algorithm. In 1981, Shiloach and Even\cite{Shiloach:1981} presented a first decremental algorithm that finds all connected components in dynamic graphs, only edges are deleted.

Henzinger and King \cite{Henzinger:1999} also proposed a new algorithm that maintains spanning tree for each connected components, which helps them to update the \ds quickly only when deletion of edge occurs.
 
The main drawback of these algorithms is, they are expensive and need more space for each change to the \ds. Also they don't utilize the advantages of multi-core or multi-processor architecture.

 In 2014, Slota, et. al., proposed a parallel multistep based algorithm using both BFS and coloring technique to detect the SCC in large graphs. Later they used the trimming methodology to reduce the search space of the graph to achieve better performance. Recently Bender et. al, \cite{Michael} proposed incremental algorithm to maintain the SCC of a dynamic graph.  Also Bloemen. et.al., \cite{Bloemen:2016} proposed a novel parallel on-the-fly algorithm for SCC decomposition in multi-core system, they used advantages of Tarjan's algorithm. 

Bader. et. al., \cite{Bader2009} developed a \ds known as STRINGER for dynamic graph problems. They used combination of both adjacency matrices and Compressed Sparse Row (CSR) representation of graph. And they claimed that the STRINGER helps faster insertions and better spatio-temporal locality as compare to the adjacency lists representations. Later they also developed a CUDA version of the STRINGER called that cuSTRINGER\cite{Green2016}, which supports dynamic graph algorithms for GPUs. 

None of above proposed algorithms clarify how the internal share-memory access is achieved by the multi-threads/processors and how the memory is synchronized, whether the \ds is linearizable or not, etc. In this paper we able to address these problems.

The rest of the paper is organized as follows. In the \secref{System-Model-Preliminaries}, we define the system model, preliminaries and design principles. In \secref{overview} we define the high level overview of the algorithm \adde \& \reme. We define the \ds of SCC-graph in \secref{const-scc-graph} and in the \secref{algorithms} we define technical details of all our algorithms and some of the pseudo-codes. In the \secref{correcteness-proof} we give high level correctness proof and in the \secref{perf-analysis} we analyze the experimental results. Finally we concluded in the \secref{conclusion-future-directions} along with future direction and discussion.
\vspace{-0.4cm}

\section{System Model \& Preliminaries}
\label{sec:System-Model-Preliminaries}
\vspace{-2mm}
In this paper, we have considered that our system consists of fixed set of $p$ processors, accessed by a finite set of $n$ threads \emph{$T_1$, $T_2$ ...., $T_n$} that run in a completely asynchronous manner and communicate through shared objects on which they perform atomic \emph{read}, \emph{write}, \emph{fetch-and-add}(FAA) operations and \emph{lock() \& \emph{unlock()}}. A FAA operation takes two arguments (loc, incVal), where \texttt{loc} is the address location from  where it fetches the value, then adds \texttt{incVal} to it and then writes back to the result \texttt{loc}.

We assume that each thread has a unique identifier, it is assigned at the time of thread creation. Each thread invokes a method which may be composed of shared-memory objects and local cipherings. We make no assumptions about the relative speeds of the threads and assume none of these processors and threads fail. 

As we said earlier our proposed algorithms are implementations of shared objects and a share object is an abstraction set of methods defined as \texttt{SCC} class in the Section \ref{sec:const-scc-graph}. It has set of methods and each method has its sequential specification. To prove a concurrent data structure to be correct, \textit{\lbty} proposed by Herlihy \& Wing \cite{HerlWing:1990:TPLS} is the standard correctness criterion. Anytime a thread invokes a method for an object, it follows until it receives a response. In may be the case a method's invocation is pending if has not received a response. For any sequential history in which the \mth{s} are ordered by their \lp{s}. 

\noindent \textbf{Progress:} An execution is \emph{deadlock\text{-}free} if it guarantees minimal progress in every \emph{crash\text{-}free}\cite{opodis_Herlihy} execution, and maximal progress if it is starvation-free. An execution is \emph{crash\text{-}free} if it guarantees
minimal progress in every uniformly isolating history, and maximal progress in some such history\cite{opodis_Herlihy}.


\noindent \textbf{Design Principles}: We developed a set of correct behaviour for our algorithm and implementation.
\begin{enumerate}
    \item \textit{thread-safety}: The SCC-graph \ds can be shared by fixed number of multiple threads at all times, which ensures all fulfill their requirement specifications and behave properly without unintended interaction.
    \item \textit{lock-freedom}: apply non-blocking techniques to provide an implementation of thread-safe C++ dynamic array based on the current C++ memory model.
    \item \textit{portability}: Generally our algorithms do not rely on specific hardware architectures, rather it is based on asynchronous memory model.
    \item \textit{simplicity}: The algorithm keeps the implementation simple to allow the correctness verification, like linearizability or model-based testing.
\end{enumerate}

\noindent \textbf{Notations:}
We denoted $\downarrow$, $\uparrow$ as input and output arguments to each method respectively and  our pseudo-code is mixed of C++ and JAVA language format.

\section{An Overview of the Algorithm}
\label{sec:overview}
\vspace{-0.2cm}
Before getting into the technical details of the algorithm, we first provide an overview of the design. The SCC class supports some basic operations: \addv, \adde, \reme, checkSCC, blongsTo, etc. and all of these methods are dead-lock free. The high-level overview of the \adde and \reme  methods are given bellow and the technical details are in the Section \ref{sec:algorithms}.
\\
\noindent
\textbf{\adde(u, v)}:
\vspace{-0.2cm}
\begin{enumerate}
    \item Checks the presence of vertices $u$, $v$ and edge($u,v$) in the SCC-Graph. If both vertices are present \& the edge is not present, adds $v$ in the $u$'s edge list and adds -$u$ in the $v$'s edge list, else returns false.
    \item After adding the edge successful, checks the \texttt{ccid} of both the vertices.
    \item If \emph{u.ccid} is same as \emph{v.ccid}, returns true, as no changes to the current SCC, else goto step 4.
    \item Checks the reachability path from vertex $v$ to $u$, if it is true, goto step 5, else returns true, as no changes to the current SCC.
    \item Runs the limited version of Tarjan's algorithm, process the affected SCCs along with its vertices and edges, merge them all to create a new SCC.
    \begin{itemize}
        \item At first it creates a new scc with any one old vertex, later adds rest of vertices to that newly created SCC and then disconnects from old SCC.
    \end{itemize} 
\end{enumerate}
\noindent
\vspace{-2mm}
\textbf{\reme(u, v)}:
\begin{enumerate}
    \item Checks the presence of vertex $u$, $v$ and edge($u,v$) in the SCC-Graph. If both are present \& edge is  present, removes $v$ from the $u$'s edge list and removes -$u$ from the $v$'s edge list, else returns false.
    \item After successful deleting the edge, checks the \texttt{ccid} of both the vertices.
    \item if \emph{u.ccid} is not same as \emph{v.ccid}, returns true, as no changes to the current SCC. Else goto step 4.
    \item Runs the forward and backward DFS algorithm(the limited version of Kosaraju’s algorithm), process all the affected vertices belongs to that SCC and creates new SCCs. 
    \begin{itemize}
        \item For each new iteration of affected vertices.
        \begin{itemize}
        \item Creates a new scc with any one of the old vertex belongs to it, later adds rest of vertices to that newly created SCC and then disconnects it from the old SCC.
    \end{itemize} 
    \end{itemize} 
\end{enumerate}

\vspace{-0.4cm}
\section{Construction of SCC-Graph structure}
\label{sec:const-scc-graph}
\vspace{-0.2cm}
In this section we present the node structures of vertex, edge and scc to construct the SCC graph. The node structures are all based on the same basic idea of lazy set implementation using linked list. This \ds is designed similarly based on the adjacency list representation of any graph. It is implemented as a collection (list) of SCCs, wherein each SCC holds the list of vertex set belongs to it, and each vertex holds the edge list (both incoming and outgoing edges). We represent all incoming edges with negative sign followed by \texttt{val} and outgoing edges with the \texttt{val}, as shown in the Fig \ref{fig:sccgraph}b.

The \texttt{Gnode} structure(similar as \cite{PeriSS16}) is a normal node and has five fields. The \texttt{val} field is the actual value of the node. If it is a vertex node, it stores the vertex id, if it is an outgoing edge, it stores the \texttt{val} of the destination vertex, if it is an incoming edge, it stores the negative of source vertex's \texttt{val}. The main idea of storing both incoming and outgoing edges for each vertex helps to explore the graph backward and forward manner respectively. And also it helps to trim the SCC-Graph after deleting a vertex, i.e, once a thread successfully deleted a vertex, all its incoming and outgoing edges needs to be removed quickly instate of iterating over whole SCC-Graph. The vertex and edge nodes are sorted in the \texttt{val} (lower to higher) order, it provides an efficient way to search when an item is absent. The boolean \texttt{marked} field is used to set the node and helps traversal to the target node without lock, we maintain an invariant that every unmarked node is reachable from the sentinel node \texttt{Head}. If a node is marked, then that is not logically present in the list. 
Each node has a \texttt{lock} field, that helps to achieve the \emph{fine\text{-}grained}  concurrency. Each node can be locked by invoking \texttt{lock()} and \texttt{unlock()} methods.It just a fine-grained locking technique, helps multiple threads can traverse the list concurrently. The \texttt{vnext} \& \texttt{enext} fields are the atomic references to the next vertex  node in the vertex list and the next edge node in the edge list of a vertex respectively.

\begin{minipage}{\linewidth}
      \centering
\begin{minipage}{.33\textwidth}
 \begin{verbatim}
unsigned long ccid; 
unsigned long ccCount
typedef struct Gnode{
  long val;   
  bool marked; 
  Lock lock;
  struct Gnode *vnext;
  struct Gnode *enext; 
}slist_t;
typedef struct CCnode{
  long ccno;   
  bool marked;
  Lock lock;
  struct Gnode* vnext;
  struct CCnode *next;
}cclist_t
 class SCC{
  CCnode CCHead, CCTail;
  bool AddVertex(u);
  bool RemoveVertex(u)
  bool AddEdge(u, v);
  bool RemoveEdge(u, v);
  bool checkSCC(u,v); 
  int  blongsTo(v); 
};
 \end{verbatim}
\end{minipage}
\hspace{0.04\linewidth}
\begin{minipage}{.55\textwidth}
\begin{figure}[H]
\centering
\captionsetup{font=scriptsize}
  \centerline{\scalebox{.90}{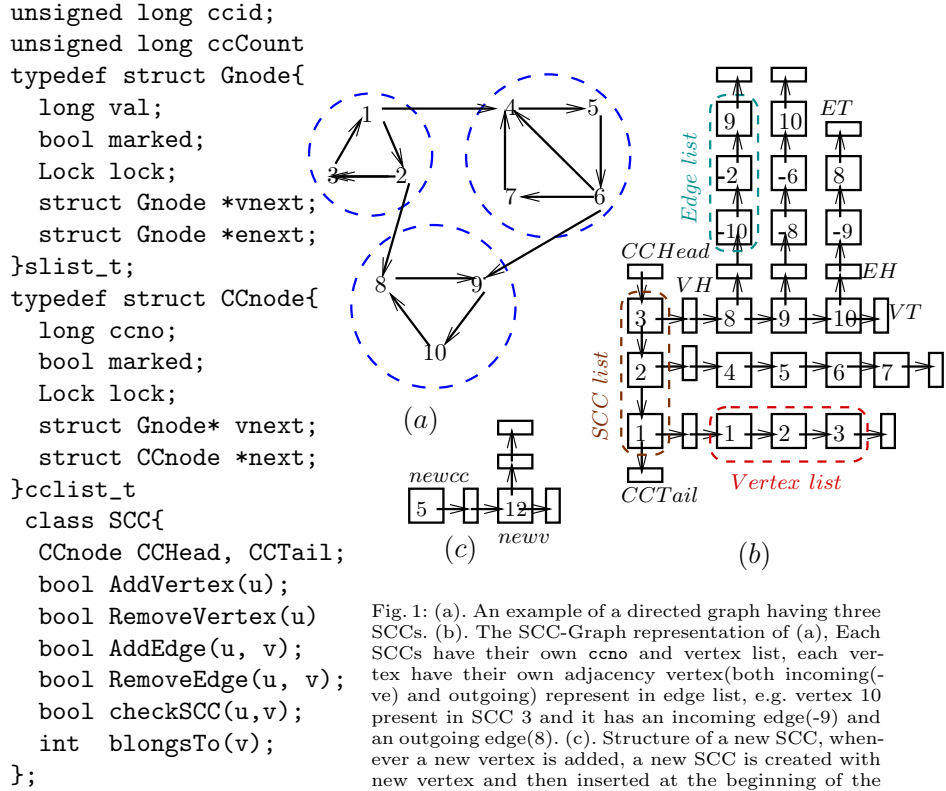}}
   \caption{(a). An example of a directed graph having three SCCs. (b). The SCC-Graph representation of (a), Each SCCs have their own \texttt{ccno} and vertex list, each vertex have their own adjacency vertex(both incoming(-ve) and outgoing) represent in edge list,  e.g. vertex $10$ present in SCC $3$ and it has an incoming edge($\text{-}9$) and an outgoing edge($8$). (c). Structure of a new SCC, whenever a new vertex is added, a new SCC is created with new vertex and then inserted at the beginning of the \texttt{CCHead} in the SCC-Graph.}
   \label{fig:sccgraph}
\end{figure}
\end{minipage}
\end{minipage}

\ignore{

\begin{figure}[ht]
   \centering
   \captionsetup{font=scriptsize}
   \centerline{\scalebox{.90}{\input{SCCFigs/sccgraphaddcc.pdf_t}}}
   \caption{(a). Example of a directed graph having three SCCs. (b). The SCC-Graph representation of (a), Each SCCs have their own \texttt{ccno} and vertex list, each vertex have their own adjacency vertex(both incoming(-ve) and outgoing(+ve)) represent in edge list. e.g. vertex $10$ present in SCC $3$ \& has one incoming edge($\text{-}9$) and one outgoing edge($8$). (c). Structure of new SCC:  whenever a new vertex is added, a new SCC is created with new vertex and then inserted at the beginning of the SCC-Graph.}
   \label{fig:sccgraph}
\end{figure}

\begin{figure}[ht]
   \centering
   \includegraphics[width=8cm, height = 5.5cm]{SCCFigs/graph}
   \caption{An example of a directed acyclic graph in the shared memory which is being accessed by multiple threads. Thread $T_1$ is trying to add a vertex $10$ to the graph. Thread $T_2$ is concurrently invoking a remove vertex $3$. Thread $T_3$ is also concurrently performing an addition of directed edge from vertex $9$ to vertex $8$ and will later create a cycle.}
   \label{fig:cloud}
\end{figure}

\begin{multicols}{2}
 \begin{verbatim}
unsigned long ccid; 
unsigned long ccCount
typedef struct Gnode{
    long val;   bool marked; 
    Lock lock;
    struct Gnode *vnext, *enext; 
}slist_t;
typedef struct CCnode{
    long ccno;   bool marked;
    Lock lock;
    struct Gnode* vnext;
    struct CCnode *next;
}cclist_t
 class SCC{
    CCnode CCHead, CCTail;
    bool checkSCC(u,v); 
    int  blongsTo(v); 
    bool AddVertex(u);
    bool RemoveVertex(u)
    bool AddEdge(u, v);
    bool RemoveEdge(u, v);
};
 \end{verbatim}
\end{multicols} 
}
The \texttt{CCnode} structure is used for holding all vertices belonging to a SCC. Like \texttt{Gnode}, it has five fields. The \texttt{ccno} field is the actual scc key value and unique for each SCC. Once a key assigned to a SCC, same key will never generate again. We assume our system provides infinite number of unique key and had no upper bound. The boolean \texttt{marked} and \texttt{lock} have same meaning as \texttt{Gnode}. The \texttt{vnext} and \texttt{next} fileds are the atomic references to vertex head(VH) and next \texttt{CCnode}.
We have two atomic variables \texttt{ccid} and \texttt{ccCount} used to hold the unique id for each CCnode and total number of SCCs respectively. 


Finally the \texttt{SCC} class is the actual abstract class, which coordinates all operation activities. This class uses two type of nodes, \texttt{Gnode} and \texttt{CCnode}. The vertex and edge nodes are represent by Gnode and the SCC nodes are represented by CCnode and also has two sentinel nodes \texttt{CCHead} and \texttt{CCTail}. The \emph{SCC} class supports four basic graph operations \emph{\addv, \adde} and \emph{\remv, \reme}, and also supports some application specific methods, \emph{\checkscc, \blongsto}, etc. The detail working and pseudo code is given in the next section. 

\vspace{-0.6cm}
\section{Algorithms}
\label{sec:algorithms}
\vspace{-0.3cm}
\noindent 
In this section we present SMSCC, the actual algorithm for maintaining strongly connected components of fully dynamic directed graph in a shared memory system. The edges and vertices are added/removed concurrently by fixed set of threads. In \secref{overview} we discussed the high level overview of two methods. The technical details of all the methods are discussed here. 

\ignore{
\begin{itemize}
    \item \textit{\adde(u, v)}: Creates a new edge if not present earlier and then checks whether the SCC is affected, if it is, tries to restore the SCC, if unable to do that returns false.
    \item \textit{\reme(u, v)}: Tries to delete the edge if it is already present and then checks whether the SCC is affected, if it is, try to restore the SCC, if unable to do that returns false.
    \item \textit{\addv(u)}: Adds a new SCC having the vertex $u$ to the SCC-graph, if the vertex is not present earlier. We assume there is no duplicate vertex in the SCC-graph, means once a vertex is added, same vertex will never be invoked by \addv method on this key again.
    \item \textit{\remv(u)}: Removes a vertex $u$ in a SCC from the SCC-graph, if the vertex is present earlier. Then checks whether the SCC is affected, if it is, try to restore the SCC, if unable to do that returns false. After that it deletes all outgoing and incoming edges.
    \item \textit{\checkscc(u, v)}: Checks whether $u$ and $v$ are in the same strongly connected component at a given instance. This method used for different applications, to check whether two he persons are present in the same community or not.
    \item \textit{\blongsto(u)}: returns the \emph{ccno} of an SCC, in which $u$ is belongs to. This method is used for applications, to find the community of a person.
  \end{itemize}
}
\vspace{-0.3cm}
\subsection{Edge or Vertex Insertion}
\label{subsec:adde}
\begin{minipage}{\linewidth}
      \centering
\begin{minipage}{.71\textwidth}
After inserting an edge to the graph, how quickly we update the SCC instate of starting everything from the scratch. The details of the algorithm is given bellow. If we allowed only insertion of edges or vertices, called it as an incremental algorithm. For this we used the modified version of Tarjan's\cite{HopcroftTarjan:1973} algorithm to restore the affected SCC after inserting an edge iff it violates the SCC-Graph. Whereas \addv will not affect the SCC-Graph. 

\hspace{4mm}To add an edge, we invoke the $\adde(u, v)$ method presented in the Algorithm \ref{alg:adde}. First it checks the presence of vertices $u$ and $v$ by invoking the \emph{\loccc} method (Algorithm \ref{algo:locatescc}) from \lineref{adde-1ftloc} to \ref{lin:adde-endifloc}. If any one of these

\end{minipage}
\hspace{0.04\linewidth}
\begin{minipage}{.22\textwidth}
\begin{figure}[H]
\centering
\captionsetup{font=footnotesize}
   \centerline{\scalebox{0.50}{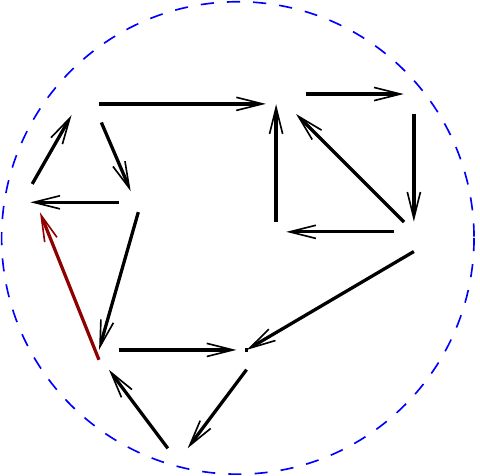}}
   \caption{An example after addition of the edge $(8,3)$ in the Fig \ref{fig:sccgraph}(a)}
   \label{fig:sccadde}
\end{figure}
\end{minipage}
\end{minipage}

\noindent
vertices is not present or the edge is present, we simply return false. After successful check of $u$ and $v$, in the \lineref{adde-addeuv} we try to add the edge node $v$(outgoing edge) in the $u$'s edge list and the edge node $\text{-}u$(incoming edge) in the $v$'s edge list. After successful addition of both the edges, we check if any changes to the SCC-Graph. For that, first we check the \texttt{ccno} of both the vertices(in the \lineref{adde-checkccno}). If they are equal, no changes to the SCC-Graph as it is added within a SCC, and we return true. On the other hand we check if there is a reachable path from vertex $v$ to $u$, if it is, we merge all the SCCs which are in the reachable path by invoking the method \emph{\findsccadde()}(in the \lineref{adde-findscc}). If the reachable path is not exist, the SCC remain unchanged, we simply return true. \\
\noindent While merging all the affected SCCs by invoking \emph{\findsccadde()} method, we only consider the vertices and edges which are affected due to addition of edge $(u, v)$. We used the modified version of the Tarjan's algorithm \cite{HopcroftTarjan:1973} because all the vertices which are in the reachable path are pushed to the stack \emph{Stk} in one iteration and then popped all to build a new SCC in one iteration as well. When the \findsccadde method is called, it creates a local stack \texttt{Stk} and other variables for processing the Tarjan's algorithm and merging the affected SCCs to single SCC. In the process of merging, for the first popped vertex , say $x$, we create a new SCC, say \texttt{newcc}(with its edges), and add it at the beginning of the SCC-Graph and then disconnect it from the old SCC, this is done by the method \emph{\createoldvcc()}, which is invoked in the \lineref{mergscc:createoldvcc}, it is similar to the \emph{\addvcc()} method(Algorithm \ref{algo:addvcc}). From the second popped vertex onwards, we just add the vertex(with its edges) to the \texttt{newcc} to the sorted position and detached the link from the old SCC, this is done by the method \emph{\addvncc()}, which is invoked in the \lineref{mergscc:addvncc}. Any time we are inserting or detaching a vertex from the list we validate as some other threads concurrent may add or delete to the predecessor or successor of that vertex. We always maintain a invariant that any unmarkable node is always reachable from the respective \texttt{Head} of the list. In the Fig \ref{fig:sccadde}, we have shown an example after addition of the edge $(8,3)$ in the Fig \ref{fig:sccgraph}(a), how all three SCCs are merged to form a new SCC. \\
\noindent A new vertex \texttt{newv} is added by invoking \emph{\addv} method. Each time this method is called with new vertex id, which is generated from the last vertex id plus one. This increment is done by atomic operation fetch-and-add (FAA). We assume all vertices have unique id and the system has unbounded number of such keys, once it is added to SCC-Graph, will never assign this id to any other vertex. Each time a new SCC also created with new \texttt{ccno}, say \texttt{newcc}. After that \texttt{newv} is added to \texttt{newcc} and then \texttt{newcc} is inserted at the beginning of the SCC-Graph and it never affects the properties of SCC-Graph. The structure of \texttt{newcc} is shown in the Fig \ref{fig:sccgraph}(c).
\vspace{-0.3cm}
\subsection{Edge or Vertex Deletion}
\label{subsec:reme}
\begin{minipage}{\linewidth}
      \centering
\begin{minipage}{.71\textwidth}
Like \adde,  after deleting an edge or a vertex from the graph, how quickly we update the SCC. If we allowed only deletion of edges or vertices to the SCC-graph, called it as decremental algorithm. We used the limited version of Kosaraju's \cite[Chap 22, Sec 22.5]{Cormen} algorithm to resort the affected SCC after deleting an edge or a vertex only iff it violates the SCC-Graph. To remove an edge, we invoke the $\reme(u, v)$ method (Algorithm \ref{alg:reme}). First we check the presence of vertices $u$ and $v$ in the SCC-Graph by invoking the \emph{\loccc} method (Algorithm \ref{algo:locatescc}) from \lineref{reme-1ftloc} to \ref{lin:reme-endifloc}. If any one of these vertices or the edge is not present, we simply return false. After successful presence of $u$ and $v$(in the \lineref{reme-remeuv}) we try to remove the edge node $v$(outgoing edge) in the $u$'s vertex list and the edge node $\text{-}u$(incoming edge) in the $v$'s vertex list, if present earlier. 

\end{minipage}
\hspace{0.04\linewidth}
\begin{minipage}{.22\textwidth}
\begin{figure}[H]
\centering
\captionsetup{font=footnotesize}
  \centerline{\scalebox{0.5}{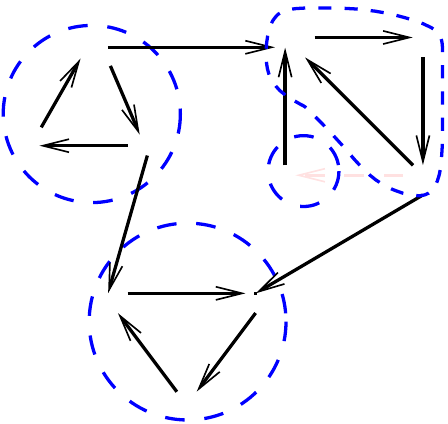}}
   \caption{An example after deletion of the edge $(8, 7)$ in the Fig \ref{fig:sccgraph}(a), the SCC  breaks to two new SCCs.}
   \label{fig:sccreme}
\end{figure}
\end{minipage}
\end{minipage}


\begin{algorithm}[H]
\captionsetup{font=footnotesize}
	\caption{It takes the input \texttt{pred} \& \texttt{curr} of type $<T>$, where $<T>$ is either a slist\_t(vertex or edge node) or cclist\_t node. It returns true with the invariant an unmarkable node is reachable from Head or else it returns false}
	\label{algo:validateCC}
	\begin{algorithmic}[1]
	\footnotesize
		\Procedure{$bool$  \val($<\text{T}>pred \downarrow, <\text{T}>curr\downarrow$ )}{}
	   \State{ return  (pred.marked = false $\wedge$ curr.marked = false $\wedge$ pred.next = curr);}
	   \EndProcedure
        \algstore{validate}
	\end{algorithmic}
\end{algorithm}
\vspace{-0.8cm}
\begin{algorithm}[H]
\captionsetup{font=footnotesize}
	\caption{It takes input as vertex \texttt{key} \& it's SCC, returns the exact location of key in vertex list of scc.}
	\label{algo:locatev}
	\begin{algorithmic}[1]
	\footnotesize
	    \algrestore{validate}
		\Procedure{\locv($key \downarrow, currc \downarrow, predv \uparrow, currv\uparrow$)}{}
	   \State{predv $\gets$ currc.vnext;}
	   \State{currv $\gets$ predv.vnext;}
	   \While{(currv.val $<$ key )}
    	   \State{predv $\gets$ currv;}
    	   \State{currv $\gets$ currv.vnext;}
       \EndWhile
       \EndProcedure
        \algstore{locv}
	\end{algorithmic}
\end{algorithm}

\vspace{-0.8cm}
\begin{algorithm}[H]
\captionsetup{font=footnotesize}
	\caption{It takes input as \texttt{key}, returns false if key is not present, else it returns true and the references to predc \& currc of SCC and the references to predv \& currv vertex having val == key.}
	\label{algo:locatescc}
	\begin{algorithmic}[1]
	\footnotesize
	    \algrestore{locv}
		\Procedure{ $bool$ \loccc($key \downarrow, predc \uparrow, currc \uparrow, predv \uparrow, currv\uparrow$ )}{}
	   \State{predc $\gets$ \cchead;}
	   \State{currc $\gets$ predc.next;}
	   \While{(currc $\neq$ \cctail)} \label{lin:loccc:whilecc}
    	   \State{predv $\gets$ currc.vnext;}
    	   \State{currv $\gets$ predv.vnext;}
    	   \While{(currv.val $\leq$ key )} 
    	       \If{(currv.val = key $\wedge$ currv.marked = false)} \label{lin:loccc:key}
    	        \State{return true;}
    	       \EndIf
    	   \State{predv $\gets$ currv;}
    	   \State{currv $\gets$ currv.vnext;}
    	   \EndWhile
    	\State{predc $\gets$ currc;}
    	\State{currc $\gets$ currc.next;}\label{lin:loccc:while:lastread}
	   \EndWhile
	   \State{return false;}
        \EndProcedure
        \algstore{locatescc}
	\end{algorithmic}
\end{algorithm}
\vspace{-0.8cm}

\begin{algorithm}[H]
\captionsetup{font=footnotesize}
	\caption{It inserts a newcc with old vertex to SCC-Graph at the CCHead position.}
	\label{algo:createoldvcc}
	\begin{algorithmic}[1]
	\footnotesize
	    \algrestore{locatescc}
		\Procedure{$bool$  \createoldvcc($predv \downarrow, currv\downarrow$)}{}
	       \State{predc $\gets$ \cchead;}
    	   \State{currc $\gets$ predc.next;}
    	   \State{newcc $\gets$ \createcc(currv$\downarrow$); }
    	   \State{predc.lock();     currc.lock();}
    	    \State{predv.lock();     currv.lock();}
    	     \If{(\valc(predc$\downarrow$, currc$\downarrow$) $\wedge$ \valv(predv$\downarrow$, currv$\downarrow$))}    
	           \State{predv.vnext $\gets$ currv.vnext}
	            \State{newcc.next $\gets$ currc;}
	            \State{predc.next $\gets$ newcc;} \label{lin:createoldv:newcc}
	            
	            \State{currv.unlock();      predv.unlock();}
	            \State{currc.unlock();      predc.unlock();}
	            \State{ccCount $\gets$ ccCount + 1;} // atomic increment (FAA(1))
	            \State{return true;}
	        \Else
	            \State{currv.unlock();      predv.unlock();}
	            \State{currc.unlock();      predc.unlock();}
	            \State{return false;}
	            
	       \EndIf
 \EndProcedure
        \algstore{createoldvcc}
	\end{algorithmic}
\end{algorithm}	   
\vspace{-0.8cm}
\begin{algorithm}[H]
\captionsetup{font=footnotesize}
	\caption{It takes input as \texttt{key} \& \ehead of a vertex, returns the exact location of key in the edge list.}
	\label{algo:locateedge}
	\begin{algorithmic}[1]
	\footnotesize
	    \algrestore{createoldvcc}
		\Procedure{\loce($key \downarrow, \ehead\downarrow, prede \uparrow, curre \uparrow$)}{}
	   \State{prede $\gets$ \ehead;}
	   \State{curre $\gets$ prede.enext;}
	   \While{(curre.val $<$ key )}
    	   \State{prede $\gets$ curre;}
    	   \State{curre $\gets$ curre.enext;}\label{lin:loce:lp}
       \EndWhile
       \EndProcedure
        \algstore{locateedge}
	\end{algorithmic}
\end{algorithm}
\vspace{-0.8cm}

\begin{algorithm}[H]
\captionsetup{font=footnotesize}
	\caption{It inserts the old vertex to SCC.}
	\label{algo:addvncc}
	\begin{algorithmic}[1]
	\footnotesize
	    \algrestore{locateedge}
		\Procedure{\addvncc($newcc\downarrow, predv \downarrow, currv\downarrow$);}{}
	      \State{flag1 $\gets$ \locv($currv.val \downarrow, newcc \uparrow, predv1 \uparrow, currv1\uparrow$);} //Algorithm \ref{algo:locatev}
	       \State{predv1.lock();        currv1.lock();}
	       \State{predv.lock();         currv.lock();}
	       \If{(flag1 = true $\wedge$ \valv(predv1$\downarrow$, currv1$\downarrow$) $\wedge$ \valv(predv$\downarrow$, currv$\downarrow$))}
	            \State{predv.vnext $\gets$ currv.vnext;} //detach the link
	           \State{currv.vnext $\gets$ currv1;} // logical insertion
	            \State{predv1.vnext $\gets$ currv;} \label{lin:oldv:newcc} // physical insertion
	            \State{currv.unlock(); predv.unlock();}
	       \State{currv1.unlock(); predv1.unlock();}
    	   \Else
    	   \State{currv.unlock();  predv.unlock();}
	       \State{currv1.unlock();predv1.unlock();}
	       \EndIf
 \EndProcedure
        \algstore{addvscc}
	\end{algorithmic}
\end{algorithm}

\vspace{-0.8cm}
\begin{algorithm}[H]
\captionsetup{font=footnotesize}
	\caption{It takes input as \emph{key} \& \ehead of a vertex. If key is not present, adds the edge node in the \ehead list, else returns false.}
	\label{algo:addedgenode}
	\begin{algorithmic}[1]
	\footnotesize
	    \algrestore{addvscc}
		\Procedure{$bool$ \addenode($key \downarrow, \ehead \downarrow$)}{}
	  \State{\loce($key \downarrow, \ehead \downarrow, prede \uparrow, curre \uparrow$);} //Algorithm \ref{algo:locateedge}
	  \State{newe $\gets$ \createe(key)}
	   \State{prede.lock(); curre.lock();}
	   \If{(curre.val $\neq$ key $\wedge$ \vale(prede$\downarrow$, curre$\downarrow$))}
    	   \State{newe.enext $\gets$ curre;} // logical insertion
    	   \State{prede.enext $\gets$ newe;} // physical insertion
    	   \State{curre.unlock(); prede.unlock();}
    	   \State{return true;}
    	\Else   
    	   \State{curre.unlock(); prede.unlock();}
    	   \State{return false;}
       \EndIf
       \EndProcedure
        \algstore{addedge}
	\end{algorithmic}
\end{algorithm}
\vspace{-0.8cm}
\begin{algorithm}[H]
\captionsetup{font=footnotesize}
	\caption{It takes input as \emph{key} \& \ehead of a vertex. If key is  present, removes the edge node from the \ehead list, else returns false.}
	\label{algo:remedgenode}
	\begin{algorithmic}[1]
	\footnotesize
	    \algrestore{addedge}
		\Procedure{$bool$ \remenode($key \downarrow, \ehead \downarrow$)}{}
	  \State{\loce($key \downarrow, \ehead \downarrow, prede \uparrow, curre \uparrow$);}
	   \State{prede.lock(); curre.lock();}
	   \If{(curre.val $=$ key $\wedge$ \vale(prede$\downarrow$, curre$\downarrow$))}
    	   \State{curre.marked $\gets$ true;} // logical deletion
    	   \State{prede.enext $\gets$ curre.enext;} // physical deletion
    	   \State{curre.unlock(); prede.unlock();}
    	   \State{return true;}
    	\Else   
    	   \State{curre.unlock(); prede.unlock();}
    	   \State{return false;}
       \EndIf
       \EndProcedure
        \algstore{remedge}
	\end{algorithmic}
\end{algorithm}

\noindent
After successful deletion, we check if any changes to SCC-Graph. For that we check the \texttt{ccno} of both the vertices(\lineref{reme-checkccno}). If the \texttt{ccnos} of both vertices are unequal, the SCC will not be affected, as edge $(u, v)$ added between two SCCs. If the \texttt{ccnos} are equal, then there may be breaking of that single SCC to multiple SCCs. For that we use the modified version of Kosaraju's algorithm for find the all SCCs in that old SCC using forward and backward depth first search(DFS) algorithm, for that we invoke the \emph{\findsccreme()} method in the \lineref{reme-findscc}.

In the Algorithm \ref{algo:findsccreme} we define to find all new SCCs after removal of the edge $(u, v)$. We only process the vertices and the edges on that SCC. For each iteration we use DFS algorithm twice, one for forward DFS invoked by \emph{\dfsf}(Algorithm \ref{alg:dfsfw} in the \lineref{findscc:dfsf}), and other invoked by \emph{\dfsb}(Algorithm \ref{alg:dfsbw}) in the \lineref{findscc:dfsb}. The \dfsf and \dfsb method locally processed the out-going and incoming edges respectively. In each iteration a new SCC is created(\emph{\createoldvcc()}) with old vertex(with its edges), say \emph{newcc} and inserted at the beginning of the SCC-Graph, and then detached from the old SCC. For the next subsequent vertices belongs to newcc's id are added by the \emph{\addvncc()} method to \texttt{newcc}. Like \adde, anytime we insert or detach a vertex from the vertex list we check the validation, because some threads concurrently may add or delete to the predecessor or successor of that vertex node. In the Fig \ref{fig:sccreme}, we have shown an example after deletion of the edge $(8, 7)$ in the Fig \ref{fig:sccgraph}(a), how a single SCC is breakdown to two new SCCs.\\
\noindent For removing a vertex we invoke the \emph{\remv} method and then use the limited version of Kosaraju's algorithm to resort the affected SCC, iff it violates the SCC-Graph. The mythology be similar as \reme, the pseudo code is shown in the Algorithm \ref{alg:remv}.

\vspace{-0.3cm}
\subsection{Check Community}
\vspace{-0.1cm}
The \emph{\checkscc(id1, id2)}, \emph{\blongsto(id)}, etc. are used for different applications. For any online social networking it takes ids of two person and checks by invoking \checkscc method whether these two persons are belong to same community, if they are, either system admin or one of the person can send friendship suggestion or request to other person. Similarly, \blongsto method reads a person id and tells which community he/she is belongs to, based on that it can do some activities in the community. Like this there are a huge number of SCC application in the dynamic graph and the requirement of efficient graph algorithms as well. The \checkscc and \blongsto methods are shown in the Algorithm \ref{alg:checkscc} and \ref{alg:blongsto} respectively.   
\vspace{-0.3cm}
\subsection{Memory management}
\vspace{-0.2cm}
Our proposed algorithm depends on a explicit garbage collector(GC) for better memory management. We defined a separate GC method which is invoked by an independent thread in regular intervals of time. Our GC method is similar to Michael’s Hazard Pointers technique\cite{Michael:2004} although it was designed for lock-free objects and we reclaim all three types of node. This GC thread does not affect the execution time. 

\begin{algorithm}[H]
\captionsetup{font=footnotesize}
	\caption{If the key is not present earlier return false, else creates a new SCC, adds a new vertex to it, returns true. Inserts the new SCC at the beginning of the SCC-graph, just after the \cchead.}
	\label{algo:addvcc}
	\begin{algorithmic}[1]
	\footnotesize
	    \algrestore{remedge}
		\Procedure{ $bool$ \addvcc($key \downarrow$ )}{}
	   \State{flag $\gets$ \loccc($key \downarrow, predc \uparrow, currc \uparrow, predv \uparrow, currv\uparrow$);} //Algorithm \ref{algo:locatescc}
	   \If{(flag = true)}
	        \State{return false;}
	   \Else
    	   \State{newcc $\gets$ \createcc(key);}
    	   \State{predc $\gets$ \cchead;}
    	   \State{currc $\gets$ predc.next;}
    	   \State{predc.lock(); currc.lock();}
	        \If{(\valc(predc$\downarrow$, currc$\downarrow$))}
	            \State{newcc.next $\gets$ currc;} // logical insertion
	            \State{predc.next $\gets$ newrc;} // physical insertion
	            \State{currc.unlock(); predc.unlock();}
	            \State{ccCount $\gets$ ccCount + 1;}
	            \State{return true;}
	        \Else
	            \State{currc.unlock(); predc.unlock();}
	            \State{return false;}
	        \EndIf
	   \EndIf
        \EndProcedure
        \algstore{addvscc}
	\end{algorithmic}
\end{algorithm}
\vspace{-0.8cm}
\begin{algorithm}[H]
\captionsetup{font=footnotesize}
	\caption{DFS of forward traversal. Only process the outgoing edges}
	\label{alg:dfsfw}
	\begin{algorithmic}[1]
	\footnotesize
	    \algrestore{addvscc}
		\Procedure{\dfsf($slHead \downarrow, sl\_edge\downarrow, num\_cc\downarrow, cc\downarrow, SUCC\downarrow \uparrow$)}{}
	   \State{SUCC[sl\_edge->val] $\gets$ num\_cc;}
	   \For{it $\gets$ sl\_edge.enext.enext to it.next $\neq$ NULL}
	   \If{(it.val $>$ 0)} // checks for outgoing edges
	    \State{flag $\gets$ \loccc($it.val \downarrow, predc \uparrow, currc \uparrow, predv \uparrow, currv\uparrow$);} //Algorithm \ref{algo:locatescc}
	   \If{(SUCC[it.val] = cc $\wedge$ flag = true)}
	   \State{\dfsf($slHead \downarrow, sl\_edge\downarrow, num\_cc\downarrow, cc\downarrow, SUCC\downarrow \uparrow$);} //Algorithm \ref{alg:dfsfw}
	   \EndIf
	   \EndIf
	   \EndFor
	   \EndProcedure
        \algstore{dfsfw}
	\end{algorithmic}
\end{algorithm}
\vspace{-0.8cm}
\begin{algorithm}[H]
\captionsetup{font=footnotesize}
	\caption{DFS of the backward traversal. Only process the incoming edges}
	\label{alg:dfsbw}
	\begin{algorithmic}[1]
	\footnotesize
	    \algrestore{dfsfw}
		\Procedure{\dfsb($slHead \downarrow, sl\_edge\downarrow, num\_cc\downarrow, cc\downarrow, PREC\downarrow \uparrow$)}{}
	   \State{PREC[sl\_edge->val] $\gets$ num\_cc;}
	   \For{it $\gets$ sl\_edge.enext.enext to it.next $\neq$ NULL}
	   \If{(it.val $<$ 0)}// checks for incoming edges
	    \State{flag $\gets$ \loccc($\text{(-1)}\ast it.val \downarrow, predc \uparrow, currc \uparrow, predv \uparrow, currv\uparrow$);} //Algorithm \ref{algo:locatescc}
	   \If{(PREC[(-1)$\ast$ it.val] = cc $\wedge$ flag = true)}
	   \State{\dfsb($slHead \downarrow, sl\_edge\downarrow, num\_cc\downarrow, cc\downarrow, PREC\downarrow \uparrow$);} //Algorithm \ref{alg:dfsbw}
	   \EndIf
	   \EndIf
	   \EndFor
	   \EndProcedure
        \algstore{dfsbw}
	\end{algorithmic}
\end{algorithm}
    \vspace{-0.8cm}
\begin{algorithm}[H]
\captionsetup{font=footnotesize}
	\caption{It adds all affected SCCs to a SCC after \adde.}
	\label{algo:findsccadde}
	\begin{algorithmic}[1]
	\footnotesize
	    \algrestore{dfsbw}
		\Procedure{$bool$ \findsccadde($scc1 \downarrow, scc2\downarrow, n\downarrow$)}{}
	   \State{Stk $\gets$ new slist\_t[n]; visited $\gets$ new bool[n];}
	   \State{Root $\gets$ new long[n]; Comp $\gets$ new long[n];}
	   \For{i $\gets$ 0 to n}
	   \State{visited[i] $\gets$ false; Root[i] $\gets$ $+\infty$; Comp[i] $\gets$ -1;}
	   \EndFor
	   \State{currv $\gets$ scc2.vnext;}
	   \State{\mergescc($currv\downarrow$, Stk$\downarrow$, Root$\downarrow$, Comp$\downarrow$, visited$\downarrow$);} // Algorithm \ref{algo:mergescc}
	   \EndProcedure
       \algstore{findsccadde}
	\end{algorithmic}
\end{algorithm}   
\vspace{-0.8cm}
\begin{algorithm}[H]
\captionsetup{font=footnotesize}
	\caption{Finds all SCCs after \reme. Iterate only affected vertices and edges. The logic is based on modified version of Kosaraju's algorithm}
	\label{algo:findsccreme}
	\begin{algorithmic}[1]
	\footnotesize
	    \algrestore{findsccadde}
		\Procedure{$bool$ \findsccreme($slHead \downarrow, cc\downarrow, n\downarrow$)}{}
	   \State{num\_cc $\gets$ ccid;}
	   \State{SUCC $\gets$ new long[n];}
	    \State{PREC $\gets$ new long[n];}
	   \For{i $\gets$ 0 to n}
	   \State{SUCC[i] $\gets$ cc;}
	   \State{PREC[i] $\gets$ cc;}
	   \EndFor
	   \For{it $\gets$ slHead.vnext to it.vnext $\neq$ NULL}
	   \If{(SUCC[it.val] $=$ cc)}
	   \State{\dfsf($slHead \downarrow, it\downarrow, num\_cc\downarrow, cc\downarrow, SUCC\downarrow \uparrow$);} \label{lin:findscc:dfsf} // Algorithm \ref{alg:dfsfw}
	   \State{\dfsb($slHead \downarrow, it\downarrow, num\_cc\downarrow, cc\downarrow, PREC\downarrow \uparrow$);}\label{lin:findscc:dfsb} // Algorithm \ref{alg:dfsbw}
	   \State{bool st $\gets$ true;} 
	   \State{cclist\_t newcc;}
	   \For{(j $\gets$ 1 to n)}
	   \If{(\loccc($key \downarrow$, $predc \uparrow$, $currc \uparrow$,$predv \uparrow$, $currv\uparrow$))} //Algorithm \ref{algo:locatescc}
	   \If{(SUCC[j] $\neq$ PREC[j])}
	     \State{SUCC[j] $\gets$ cc;} 
	     \State{PREC[j] $\gets$ cc;}
	   \Else 
	     \If{(SUCC[j] = num\_cc $\wedge$ PREC[j] = num\_cc))}
	       \If{(st = true)} // for first vertex
	       \State{newcc $\gets$ \createoldvcc($predv \downarrow$, $currv\downarrow$);}// Algorithm \ref{algo:createoldvcc}
	       \State{st $\gets$ false;} 
	       \Else // for rest of vertices with same num\_cc
	        \State{\addvncc($newcc\downarrow$, $predv \downarrow$, $currv\downarrow$);} // Algorithm \ref{algo:addvncc}
	     \EndIf
	   \EndIf
	   \EndIf
	   \EndIf
	   \EndFor
	   \State{num\_cc $\gets$ num\_cc +1;} 
	   \State{ccid $\gets$ num\_cc;}
	   \EndIf
	   
	   \EndFor
	   \EndProcedure
        \algstore{findsccreme}
	\end{algorithmic}
\end{algorithm}

\vspace{-0.8cm}
\begin{algorithm}[H]
\captionsetup{font=footnotesize}
	\caption{Finds all SCCs after \adde. Iterate only the vertices and edges which are affected and the logic is based on limited version of Tarjan's algorithm.}
	\label{algo:mergescc}
	\begin{algorithmic}[1]
	\footnotesize
	    \algrestore{findsccreme}
		\Procedure{\mergescc($currv\downarrow, Stk\downarrow, Root\downarrow, Comp\downarrow, visited\downarrow$)}{}
	   \State{visited[currv.val] $\gets$ true;}
	   \State{Root[currv.val] $\gets$ currv.val;}
	   \State{Comp[currv.val] $\gets$ -1;}
	   \State{Stk.push(currv)}
	   \For{it $\gets$ currv.enext.enext to it.enext $\neq$ NULL}
	   \If{(it.val $>$ 0)}
	   \If{(visited[it.val] = false)}
	   \If{(\loccc($it.val \downarrow, predc \uparrow, currc \uparrow, predv \uparrow, currv\uparrow$))}  //Algorithm \ref{algo:locatescc}
	   \State{\mergescc($currv\downarrow, Stk\downarrow, Root\downarrow, Comp\downarrow, visited\downarrow$);}  //Algorithm \ref{algo:mergescc}
	   \EndIf
	   \EndIf
	   \If{(Comp[it.val] = -1)}
	    \State{Root[currv.val] $\gets$ Root[currv.val] $<$ Root[it.val] $?$ Root[currv.val] $\colon$ Root[it.val]}
	   \EndIf
	   \EndIf
	   \EndFor
	   
	   \If{(Root[currv.val] = currv.val))}
	    \State{newcc $\gets$ \createoldvcc($predv \downarrow, currv\downarrow$);}  \label{lin:mergscc:createoldvcc} //Algorithm \ref{algo:createoldvcc}
	   \State{ccCount $\gets$ ccCount + 1;} // atomic increment(FAA(1))
	   \Repeat
	   \State{w $\gets$ Stk.pop();}
	   \State{\addvncc($newcc\downarrow,  w\downarrow$);}\label{lin:mergscc:addvncc}  //Algorithm \ref{algo:addvncc}
	   \Until{(w.val $\neq$ currv.val)} 
	   \EndIf
	   \EndProcedure
        \algstore{mergescc}
	\end{algorithmic}
\end{algorithm}

\vspace{-0.8cm}

\begin{algorithm}[H]
\captionsetup{font=footnotesize}
	\caption{Adds both incoming and outgoing edges to the edge list of vertex $key_1$ ,if it is not present earlier and then update the affected SCCs.}
	\label{alg:adde}
	\begin{algorithmic}[1]
	\algrestore{mergescc}
	\footnotesize
		\Procedure{$bool$ \adde ($key_1\downarrow$,$key_2\downarrow$)}{}
		 \State{flag1 $\gets$ \loccc($key_1 \downarrow, predc1 \uparrow, currc1 \uparrow, predv1 \uparrow, currv1\uparrow$);} \label{lin:adde-1ftloc}//Algorithm \ref{algo:locatescc}
		 \State{flag2 $\gets$ \loccc($key_2 \downarrow, predc2 \uparrow, currc2 \uparrow, predv2 \uparrow, currv2\uparrow$);} \label{lin:adde-2ftloc}//Algorithm \ref{algo:locatescc}
		\If {( flag1 = false $\vee$  flag2 = false)}
		\State {return false; }
        \EndIf
        \State{flag1 $\gets$ \loccc($key_1 \downarrow, predc1 \uparrow, currc1 \uparrow, predv1 \uparrow, currv1\uparrow$);} \label{lin:adde-1ftloca}//Algorithm \ref{algo:locatescc}
		\If {($flag1 = false $)}
      	 \State {return false; }
        \EndIf \label{lin:adde-endifloc}
        \State{flag $\gets$ \addenode(curr1.enext, $key_2$) $\wedge$ \addenode(curr2.enext, $(-1)*key_1$);} \label{lin:adde-addeuv} //Algorithm \ref{algo:addedgenode}
        \If{(flag = true)}
         \If{(currc1.ccno =  currc2.ccno)}
            \label{lin:adde-checkccno}
         \State{return true;}
         \Else
         \If{(!isReachable(currc2,currc1))} \label{lin:adde-reachable} // Checks reachable path from currc2 to currc1
          \State{ return true;}
         \Else
          \If{(\findsccadde($currc2 \downarrow, currc1\downarrow, n\downarrow$))}\label{lin:adde-findscc} //Algorithm \ref{algo:findsccadde}
          \State{return true;}
          \Else
          \State{return false;}
         \EndIf
        \EndIf
        \EndIf
        \EndIf
        
        \EndProcedure
		\algstore{adde}
	\end{algorithmic}
\end{algorithm}
\vspace{-0.8cm}

\begin{algorithm}[H]
\captionsetup{font=footnotesize}
	\caption{Removes both incoming and outgoing edges from the edge list of vertex $key_1$,if it is present and then update the affected SCCs.}
	\label{alg:reme}
	\begin{algorithmic}[1]
	\algrestore{adde}
	\footnotesize
		\Procedure{$bool$ \reme ($key1\downarrow$,$key2\downarrow$)}{}
		 \State{flag1 $\gets$ \loccc($key1 \downarrow, predc1 \uparrow, currc1 \uparrow, predv1 \uparrow, currv1\uparrow$);} \label{lin:reme-1ftloc} //Algorithm \ref{algo:locatescc}
		 \State{flag2 $\gets$ \loccc($key2 \downarrow, predc2 \uparrow, currc2 \uparrow, predv2 \uparrow, currv2\uparrow$);}\label{lin:reme-2ftloc} //Algorithm \ref{algo:locatescc}
		\If {( flag1 = false $\vee$  flag2 = false)}
		\State {return false; }
        \EndIf
        \State{flag1 $\gets$ \loccc($key1 \downarrow, predc1 \uparrow, currc1 \uparrow, predv1 \uparrow, currv1\uparrow$);}\label{lin:reme-1ftloca} //Algorithm \ref{algo:locatescc}
		\If {($flag1 = false $)}
      	 \State {return false; }
        \EndIf \label{lin:reme-endifloc}
        
        \State{flag $\gets$ \remenode(curr1.enext, $key2$) $\wedge$ \remenode(curr2.enext, $(-1)*key1$);} \label{lin:reme-remeuv} //Algorithm \ref{algo:remedgenode}
        \If{(flag = true)}
         \If{(currc1.ccno $\neq$  currc2.ccno)}\label{lin:reme-checkccno}
         \State{return true;}
         \Else
         \If{(\findsccreme($currc1 \downarrow, n\downarrow$))}\label{lin:reme-findscc} //Algorithm \ref{algo:findsccreme}
          \State{return true;}
          \Else
          \State{return false;}
         \EndIf
        \EndIf
        \EndIf
        \EndProcedure
		\algstore{reme}
	\end{algorithmic}
\end{algorithm}

\vspace{-0.8cm}

\begin{algorithm}[H]
\captionsetup{font=footnotesize}
	\caption{Removes the vertex along its incoming and outgoing edges from SCC-Graph,if it is present else returns false}
	\label{alg:remvn}
	\begin{algorithmic}[1]
	\algrestore{reme}
	\footnotesize
		\Procedure{$bool$ \remvnode ($key\downarrow$, $cc\uparrow$)}{}
		 \State{flag $\gets$ \loccc($key \downarrow, predc \uparrow, currc \uparrow, predv \uparrow, currv\uparrow$);} \label{lin:remvnode} //Algorithm \ref{algo:locatescc}
		 \If {($flag = false $)}
      	 \State {return false; }
      	 \Else
      	  \State{predv.lock(); currv.lock();}
	   \If{\vale(predv$\downarrow$, currv$\downarrow$))}
    	   \State{currv.marked $\gets$ true;} // logical deletion
    	   \State{predv.enext $\gets$ currv.enext;} // physical deletion
    	   \State{cc $\gets$ currc;}
    	   \State{currv.unlock(); predv.unlock();}
    	   \State{return true;}
    	\Else   
    	   \State{currv.unlock(); predv.unlock();}
    	   \State{return false;}
        \EndIf

        \EndIf
        \EndProcedure
		\algstore{remvn}
	\end{algorithmic}
\end{algorithm}

\vspace{-0.8cm}

\begin{algorithm}[H]
\captionsetup{font=footnotesize}
	\caption{update the SCC-Graph after successful \remvnode}
	\label{alg:remv}
	\begin{algorithmic}[1]
	\algrestore{remvn}
	\footnotesize
		\Procedure{$bool$ \remv ($key\downarrow$)}{}
		 \State{flag $\gets$ \remvnode($key \downarrow$, $cc \uparrow$);} \label{lin:remvn} //Algorithm \ref{alg:remvn}
		 \If {($flag = false $)}
      	 \State {return false; }
      	 \Else
      	  \State{return \findsccreme($cc \downarrow, cc.ccno)\downarrow$))}\label{lin:remv-findscc} //Algorithm \ref{algo:findsccreme};
        \EndIf
        \EndProcedure
		\algstore{remv}
	\end{algorithmic}
\end{algorithm}

\vspace{-0.8cm}

\begin{algorithm}[H]
\captionsetup{font=footnotesize}
	\caption{It inserts a newcc with new vertex to SCC-Graph at the CCHead position}
	\label{alg:addcc}
	\begin{algorithmic}[1]
	\algrestore{remv}
	\footnotesize
		\Procedure{$bool$ \addcc ($key\downarrow$)}{}
		 \State{predc $\gets$ \cchead;}
    	   \State{currc $\gets$ predc.next;}
    	   \State{newcc $\gets$ \createc(key$\downarrow$); } //Algorithm \ref{alg:createc};
    	   \State{predc.lock();     currc.lock();}
    	    \State{predv.lock();     currv.lock();}
    	     \If{(\valc(predc$\downarrow$, currc$\downarrow$) $\wedge$ \valv(predv$\downarrow$, currv$\downarrow$))}    
	           \State{predv.vnext $\gets$ currv.vnext}
	            \State{newcc.next $\gets$ currc;} // logical insertion
	            \State{predc.next $\gets$ newcc;} // physical insertion 
	            
	            \State{currv.unlock();      predv.unlock();}
	            \State{currc.unlock();      predc.unlock();}
	            \State{ccCount $\gets$ ccCount + 1;} // atomic increment (FAA(1))
	            \State{return true;}
	        \Else
	            \State{currv.unlock();      predv.unlock();}
	            \State{currc.unlock();      predc.unlock();}
	            \State{return false;}
	            
	       \EndIf
        \EndProcedure
		\algstore{addcc}
	\end{algorithmic}
\end{algorithm}
\vspace{-0.8cm}

\begin{algorithm}[H]
\captionsetup{font=footnotesize}
	\caption{Adds a new vertex to the SCC-Graph}
	\label{alg:addv}
	\begin{algorithmic}[1]
	\algrestore{addcc}
	\footnotesize
		\Procedure{$bool$ \addv ($key\downarrow$)}{}
		 \State{ return \addcc ($key\downarrow$);} // Algorithm \ref{alg:addcc}
        \EndProcedure
		\algstore{addv}
	\end{algorithmic}
\end{algorithm}

\begin{algorithm}[H]
\captionsetup{font=footnotesize}
	\caption{Initialize a new SCC}
	\label{alg:createc}
	\begin{algorithmic}[1]
	\algrestore{addv}
	\footnotesize
		\Procedure{$cclist\_t$ \createc ($key\downarrow$)}{}
		 \State{ VHead.val $\gets$ INT\_MIN;}
		 \State{ VHead.vnext $\gets$ NULL;}
		 \State{ VHead.enext $\gets$ NULL;}
		 \State{ VHead.marked $\gets$ false;}
		 \State{ VTail.val $\gets$ INT\_MAX;}
		 \State{ VTail.vnext $\gets$ NULL;}
		 \State{ VTail.enext $\gets$ NULL;}
		 \State{ VTail.marked $\gets$ false;}
		 \State{ EHead.val $\gets$ INT\_MIN;}
		 \State{ EHead.vnext $\gets$ NULL;}
		 \State{ EHead.enext $\gets$ NULL;}
		 \State{ EHead.marked $\gets$ false;}
		 \State{ ETail.val $\gets$ INT\_MAX;}
		 \State{ ETail.vnext $\gets$ NULL;}
		 \State{ ETail.enext $\gets$ NULL;}
		 \State{ ETail.marked $\gets$ false;}
		 \State{ EHead.enext $\gets$ ETail;}
		 \State{ newv.val $\gets$ key;}
		 \State{ newv.vnext $\gets$ VTail;}
		 \State{ newv.enext $\gets$ EHead;}
		 \State{ newv.marked $\gets$ false;}
		 \State{ newcc.vnext $\gets$ newv;}
		 \State{ newcc.next $\gets$ NULL;}
		 \State{ newcc.ccno $\gets$ ccid;}
		 \State{ ccid $\gets$ ccid $+$ 1;}
		 \State{ return newcc;}
        \EndProcedure
		\algstore{createc}
	\end{algorithmic}
\end{algorithm}

\begin{algorithm}[H]
\captionsetup{font=footnotesize}
	\caption{Removes all empty SCCs from SCC-Graph, i.e., SCC having empty vertex}
	\label{alg:remcc}
	\begin{algorithmic}[1]
	\algrestore{createc}
	\footnotesize
		\Procedure{$bool$ \remcc()}{}
		 \State{predc $\gets$ \cchead;}
    	   \State{currc $\gets$ predc.next;}
    	   \While{(currc $\neq$ CCTail)}\label{lin:remcc:while}
    	   \If{(currc.vnext.vnext.vnext = NULL)}
    	   \State{predc.lock();} 
    	   \State{currc.lock(); }
    	    \If{(\valc(predc$\downarrow$, currc$\downarrow$) )}    
	           \State{currc.marked $\gets$ true;} // logical deletion
	            \State{predc.next $\gets$ currc;} // physical deletion 
	            \State{currc.unlock();}      
	            \State{predc.unlock();}      
	            \State{ccCount $\gets$ ccCount - 1 ;} // atomic increment (FAA(-1))
	            \State{continue;}// goto \lineref{remcc:while};
	        \Else
	            \State{currc.unlock();}      
	            \State{predc.unlock();}  
	            \State{predc $\gets$ currc;} 
	            \State{currc $\gets$ currc.next;} 
	       \EndIf
	       \Else
	            \State{predc $\gets$ currc;} 
	            \State{currc $\gets$ currc.next;} 
    	   \EndIf
    	   \EndWhile
    	   
        \EndProcedure
		\algstore{remcc}
	\end{algorithmic}
\end{algorithm}
\vspace{-0.8cm}

\begin{algorithm}[H]
\captionsetup{font=footnotesize}
	\caption{Checks whether two ids are in the same strongly connected component at a given instance.}
	\label{alg:checkscc}
	\begin{algorithmic}[1]
	\algrestore{remcc}
	\footnotesize
		\Procedure{$bool$ \checkscc($key_1\downarrow$, $key_2\downarrow$)}{}
		\State{flag1 $\gets$ \loccc($key_1 \downarrow, predc1 \uparrow, currc1 \uparrow, predv1 \uparrow, currv1\uparrow$);} \label{lin:cone-1ftloc}//Algorithm \ref{algo:locatescc}
		 \State{flag2 $\gets$ \loccc($key_2 \downarrow, predc2 \uparrow, currc2 \uparrow, predv2 \uparrow, currv2\uparrow$);} \label{lin:cone-2ftloc}//Algorithm \ref{algo:locatescc}
		\If {( flag1 = false $\vee$  flag2 = false)}
		\State {return false; }
        \EndIf
        \State{flag1 $\gets$ \loccc($key_1 \downarrow, predc1 \uparrow, currc1 \uparrow, predv1 \uparrow, currv1\uparrow$);} \label{lin:cone-1ftloca}//Algorithm \ref{algo:locatescc}
		\If {($flag1 = false $)}
      	 \State {return false; }
      	 \Else
      	 \State{\loce($key_2 \downarrow, currv1.enext \downarrow, prede \uparrow, curre \uparrow$);} \label{lin:cone-loce}//Algorithm \ref{algo:locateedge}
      	 
		\If {($curre.val = key_2 \wedge curre.marked = false$)}
		\State{return true;}
		\Else
		\State{return false;}
        \EndIf \label{lin:cone-endifloc}
        \EndIf
        \EndProcedure
		\algstore{checkscc}
	\end{algorithmic}
\end{algorithm}
\vspace{-0.8cm}

\begin{algorithm}[H]
\captionsetup{font=footnotesize}
	\caption{Checks id belongs to which SCC, it returns the \emph{ccno} of an SCC,}
	\label{alg:blongsto}
	\begin{algorithmic}[1]
	\algrestore{checkscc}
	\footnotesize
		\Procedure{$bool$ \blongsto($key\downarrow$)}{}
		\State{flag $\gets$ \loccc($key \downarrow, predc \uparrow, currc \uparrow, predv \uparrow, currv\uparrow$);} \label{lin:conv-1ftloc}//Algorithm \ref{algo:locatescc}
		 \If {($flag = true \wedge currv.marked = false $)}
      	 \State {return true; }
      	
		\Else
		\State{return false;}
        
        \EndIf
        \EndProcedure
	\end{algorithmic}
\end{algorithm}
\vspace{-0.8cm}

\vspace{-0.3cm}
\section{The Correctness Proof}
\label{sec:correcteness-proof}
\vspace{-0.3cm}
We now describe how our proposed algorithm SMSCC is correct. A full and detail proof is incomplete in this paper and it is based on Timnat. et.al.'s, full paper \cite{TimnatBKP12}. We think the detail proof is very much important for concurrent \ds and algorithms as without that, it is very hard to understand the races. Any directed graph is represented as SCC-Graph and it is collection of three types of lists. First, the \emph{SCC-list}, each SCC is a node in the SCC-list. Secondly, \emph{Vertex-list}, each SCC has a vertex set, stoored in the vertex list and finally, \emph{Edge-list}, each vertex has its  adjacency edge list. The SCC-Graph is interfaced with node id or key value \texttt{val}, boolean \texttt{marked} filed and \texttt{next} field. At any instance of time a node is considered to be part of SCC-Graph, if it is unmarked. \\

\noindent
\textbf{Proof Methodology}
We define the abstract SCC-Graph which always holds two invariants. Once the invariant holds for a node, it remain true. The first invariant is that, the node(SCC or vertex or edge) can only physically change by pointer(\texttt{next} or \texttt{vnext} or \texttt{enext}) and the key value of the node never change after initialization. Second, once a node is marked, it remain to be marked and it's next pointer never change until GC. For proving the correctness we use the four stages of any node similar like Timnat. et.al.'s,  \cite{TimnatBKP12}. \emph{Logical remove}: changing the \texttt{marked} filed false to true. \emph{Physical remove}: delinking the node from the list. \emph{Logical insertion}: Connecting new node's pointer to the node list. \emph{Physical Insertion}: making new logical node to a physical node, i.e. actual insertion.  We prove our algorithm using mathematical induction. 
\begin{lemma}
\label{lem:linealizability-g}
The history $H$ generated by the interleaving of any of the methods of the SCC-Graph, is \lble. \\
\textnormal{Proof is incomplete}.
\end{lemma}
\begin{lemma}
	The methods $\addv$, $\remv$, $\adde$ and $\reme$ are deadlock-free.
\end{lemma}
\begin{proofsketch}
\noindent
We prove all the $\addv$, $\remv$, $\adde$ and $\reme$ methods(based on \cite{PeriSS16}) are deadlock-free by direct argument based of the acquiring lock on both the current and predecessor nodes.
\begin{enumerate}
    \item $\addv$: the $\addv(key)$ method is deadlock-free because a thread always acquires lock on the $\vnode$ with smaller keys first. Which means, if a thread say $T_1$ acquired a lock on a $\vnode(key)$, it never tries to acquire a lock on a $\vnode$ with key smaller than or equal to $\vnode(key)$. This is true because the $\addv$ method acquires lock on the predecessor $\vnode$ from the $LocateVertex$ method. 
    \item $\remv$: the $\remv(key)$ method is also deadlock-free, similar argument as $\addv$. 
    \item $\adde$: the $\adde(key_1, hey_2)$ method is deadlock-free because a thread always acquires lock on the $\enode$ with smaller keys first. Which means, if a thread say $T_1$ acquired a lock on a $\enode(key_2)$, it never tries to acquire a lock on a $\enode$ of the vertex $\vnode(key_1)$ with key smaller than or equal to $\enode(key_2)$. This is true because the $\adde$ method acquires lock on the predecessor edge nodes of the vertex $\vnode(key_1)$ from the $\loce$ method. 
    \item $\reme$: the $\reme(key_1, key_2)$ method is also deadlock-free, similar argument as $\adde$. 
\end{enumerate}

\end{proofsketch}
\begin{lemma}
	The methods $\checkscc$ and $\blongsto$ are wait-free.
\end{lemma}

\begin{proof}
(Based on \cite{PeriSS16})The $\blongsto(key)$ method scans the vertex list of the graph starting from the $\vh$, ignoring whether $\vnode$ are marked or not. It returns a boolean $flag$ either $true$ or $false$ depending on $\vnode(key)$ greater than or equal to the sought-after key. If the desired $\vnode$ is unmarked, it simply returns $true$ and this is correct because the vertex list is sorted. On the other hand, it returns $false$ if $\vnode(key)$ is not present or has been marked. This $\blongsto$ method is wait-free, because there are only a finite number of vertex keys that are smaller than the one being searched for. By the observation of the code, a new $\vnode$ with lower or equal keys is never added ahead of it, hence they are reachable from $\vh$ even if vertex nodes are logically removed from the vertex list. Therefore, each time the $\blongsto$ moves to a new vertex node, whose key value is larger key than the previous one. This can happen only finitely many times, which says the traversal of $\blongsto$ method is wait-free.

Similarly, the $\checkscc(key_1, key_2)$ method first scans the vertex list of the graph starting from the $\vh$, ignoring whether vertex nodes are marked or not. It returns a boolean $flag$ either $true$ or $false$ depending on $\enode(key_2)$ greater than or equal to the sought-after key in the edge list of the vertex $\vnode(key_1)$. If the desired $\enode$ is unmarked, it simply returns $true$ and this is correct because the vertex list is sorted as well as the edge list of the vertex $\vnode(key_1)$ is also sorted. On the other hand it returns $false$ if either $\vnode(key_1)$ or $\vnode(key_2)$ is not present or has been marked in the vertex list or $\enode(key_2)$ is not present or has been marked in the edge list of the vertex $\vnode(key_1)$. This $\checkscc$ method is wait-free, because there are only a finite number of vertex keys that are smaller than the one being searched for as well as a finite number of edge keys that are smaller than the one being searched for in edge list of any vertex. By observation of the code, a new $\enode$ with lower or equal keys is never added ahead of $\enode(key_2)$ in the edge list of the vertex $\vnode(key_1)$, hence they are reachable from $\vh$ even if vertex nodes or edge nodes of $\vnode(key_1)$ are logically removed from the vertex list. Therefore, each time the  $\cone$ moves to a new edge node, whose key value is larger key than the previous one. This can happen only finitely many times, which says the traversal of $\checkscc$ method is wait-free. 

\end{proof}
\vspace{-0.6cm}
\subsection{Linearization Points}
\vspace{-0.1cm}
In this section we identify the linearization point(LP) of our proposed methods and it is similar as \cite{PeriSS16}. Before identifying the LP, we first consider \loccc, as it is used by most of the methods. It returns true if key value is present along with pair of SCC pointers(\texttt{predc} \& \texttt{currc}) and pair of vertex pointers (\texttt{predv} \& \texttt{currv}). For successful \emph{\loccc} return the LP be \lineref{loccc:key} where key is found and for unsuccessful the LP be \lineref{loccc:while:lastread} last read of \texttt{currc.next}. \\
\noindent
The LP of a successful \emph{\adde} with no successful concurrent \emph{\remv}, is either in the \lineref{createoldv:newcc} when a new SCC is created with the first old vertex, it is the physical insertion \texttt{predc.next $\gets$ newcc} at the \texttt{CCHead} position. Or the LP be in the \lineref{oldv:newcc} \texttt{predv1.vnext $\gets$ currv} physical insertion of old vertex to \texttt{newcc}. If there is a successful concurrent removal of either one of the vertex or both, we linearized just before the LP of the first successful concurrent \remv. For unsuccessful \emph{\adde}, the LP is inside \emph{\loccc} method where either of the vertex is not found in the \lineref{loccc:while:lastread} the last read of \texttt{currc.next} or inside the \emph{\loce} method in the \lineref{loce:lp}, the last read of \emph{curre.enext} inside the while loop, if the edge node is present.\\ 

\noindent
Similarly, the LP of a successful \emph{\reme} with no successful concurrent \emph{\remv} is same as successful \emph{\adde} with no successful concurrent \emph{\remv}. If there is a successful concurrent removal of either one of the vertex or both, we linearized just before the LP of the first successful concurrent \remv. For unsuccessful \emph{\adde}, the LP is inside \emph{\loccc} method where either of the vertex key is not found in the \lineref{loccc:while:lastread} the last read of \texttt{currc.next}, or inside the \emph{\loce} method in the \lineref{loce:lp}, the last read of \emph{curre.enext} inside the while loop, where the edge node is not present. 

\vspace{-0.3cm}
\section{Performance Analysis}
\label{sec:perf-analysis}
\vspace{-0.3cm}
In this section, we evaluate the performance of our SMSCC algorithm. The source code available at https://github.com/Mukti0123/SMSCC. It contains both fully and partial dynamic SCC with \& without deletion of incoming edges(DIE) and some applications, such as community detection. We  compare throughput with sequential and coarse-grain. \\

The methods are evaluated on a dual-socket, 10 cores per socket, Intel Xeon (R) CPU E5-2630 v4 running at 2.20 GHz frequency. Each core supports 2 hardware threads. Every core's L1 has 64k, L2 has 256k cache memory are private to that core; L3 cache (25MB) is shared across all cores of a processors.
All the codes are compiled using the GCC C/C++ compiler (version 5.4.0) with -O3 optimization and Posix threads execution model. \\
\noindent\textbf{Workload \& methodology:}
we ran each experiment
for 20 seconds, and measured the overall number of operations executed by all the
threads(starting from 1, 10, 20 to 60). The graphs shown in the Fig \ref{fig:scc} \& \ref{fig:scc-app} are the total number of operations executed by all threads. In all the tests, we ran each evaluation 8 times and took the average.\\
The algorithms we compare are, 
(1). Sequential(only one thread and no lock) with partial (without removing vertices, \emph{Seq-woDV}) and fully(\emph{Seq}) dynamic, (2). Coarse-grained(only one spin lock) with partial (without removing vertices, \emph{Coarse-woDV}) and fully(\emph{Coarse}) dynamic, (3). SMSCC with partial (without removing vertices, \emph{SMSCC-woDV}) and fully(\emph{SMSCC}) dynamic(with and without DIE). Each thread performed, in the Fig \ref{fig:scc-50-50}, 50\% add(V+E) and 50\% rem(V+E), in the Fig \ref{fig:scc-90-10}, 90\% add(V+E) and 10\% rem(V+E)  and 10\% add(V+E) and in the Fig \ref{fig:scc-10-90}, 90\% rem(V+E). The Fig \ref{fig:scc-app} shows the throughputs of \emph{Seq}, \emph{Coarse}, \emph{SMSCC}(without DIEs) and \emph{SMSCC-DIE}(with DIEs). Similarly, the Fig \ref{fig:scc-inc} and \ref{fig:scc-dec} depicts the incremental(SMISCC) \& decremental(SMDSCC) throughputs respectively and Fig \ref{fig:scc-com} shows the community detection(80\% check and 20\% add \& rem). 

After executing all above micro benchmarks, SMSCC (with \& without DIEs) perform efficiently over Sep and Coarse. The Fig \ref{fig:scc} and \ref{fig:scc-app} shows the performance is similar to the lazy linked list and the throughput is increased between $3$ to $7$X depending on different workload distributions and applications. 

\begin{figure}
\setlength{\belowcaptionskip}{-10pt}
\captionsetup{font=scriptsize}
    \begin{subfigure}[b]{0.32\textwidth}
    \setlength{\belowcaptionskip}{-7pt}
    \captionsetup{font=scriptsize}
    \caption{Add 50\% \& Rem 50\%}
        \centering
        \resizebox{\linewidth}{!}{
	\begin{tikzpicture} [scale=0.450]
	\begin{axis}[legend style={at={(0.5,1)},anchor=north},
	xlabel=No of threads,
	ylabel=Throughput ops/sec, ylabel near ticks]
	\addplot table [x=Threads, y=$Seq-woDV$]{results/scc5050.dat};
	\addplot table [x=Threads, y=$Seq$]{results/scc5050.dat};
	\addplot table [x=Threads, y=$coarse-woDV$]{results/scc5050.dat};
	\addplot table [x=Threads, y=$coarse$]{results/scc5050.dat};
	\addplot table [x=Threads, y=$SMSCC-woDV$]{results/scc5050.dat};
	\addplot table [x=Threads, y=$SMSCC$]{results/scc5050.dat};
	\end{axis}
	\end{tikzpicture}
        }
        \label{fig:scc-50-50}
    \end{subfigure}
    \begin{subfigure}[b]{0.32\textwidth}
    \setlength{\belowcaptionskip}{-7pt}
    \captionsetup{font=scriptsize}
    \caption{Add 90\% \& Rem10\%}   
    \centering
        \resizebox{\linewidth}{!}{
   	\begin{tikzpicture}[scale=0.450]
	\begin{axis}[legend style={legend columns=4,{at={(1.1,-0.2)}}},
	xlabel=No of threads,
	ylabel=Throughput ops/sec, ylabel near ticks]
	\addlegendentry{$Seq\text{-}woDV$}
	\addplot table [x=Threads, y=$Seq-woDV$]{results/scc9010.dat};
	\addlegendentry{$Seq$}
	\addplot table [x=Threads, y=$Seq$]{results/scc9010.dat};
	\addlegendentry{$Coarse\text{-}woDV$}
	\addplot table [x=Threads, y=$coarse-woDV$]{results/scc9010.dat};
	\addlegendentry{$Coarse$}
	\addplot table [x=Threads, y=$coarse$]{results/scc9010.dat};
	\addlegendentry{$SMSCC\text{-}woDV$}
	\addplot table [x=Threads, y=$SMSCC-woDV$]{results/scc9010.dat};
	\addlegendentry{$SMSC$}
	\addplot table [x=Threads, y=$SMSCC$]{results/scc9010.dat};
	\end{axis}
	\end{tikzpicture}
        }
        \label{fig:scc-90-10}
    \end{subfigure}
    \begin{subfigure}[b]{0.32\textwidth}
    \captionsetup{font=scriptsize}
    \setlength{\belowcaptionskip}{-7pt}
    \caption{Add 10\% \& Rem90\%}
        \centering
        \resizebox{\linewidth}{!}{
           \begin{tikzpicture}[scale=0.450]
	\begin{axis}[legend style={at={(0.5,0.6)},anchor=north},
	xlabel=No of threads,
	ylabel=Throughput ops/sec, ylabel near ticks]
    \addplot table [x=Threads, y=$Seq-woDV$]{results/scc1090.dat};
	\addplot table [x=Threads, y=$Seq$]{results/scc1090.dat};
	\addplot table [x=Threads, y=$coarse-woDV$]{results/scc1090.dat};
	\addplot table [x=Threads, y=$coarse$]{results/scc1090.dat};
	\addplot table [x=Threads, y=$SMSCC-woDV$]{results/scc1090.dat};
	\addplot table [x=Threads, y=$SMSCC$]{results/scc1090.dat};
	\end{axis}
	\end{tikzpicture}
        }
        \label{fig:scc-10-90}
    \end{subfigure}
        \vspace{-5mm}
\caption{SMSCC Execution with different workload distributions} 
\label{fig:scc}
\end{figure}
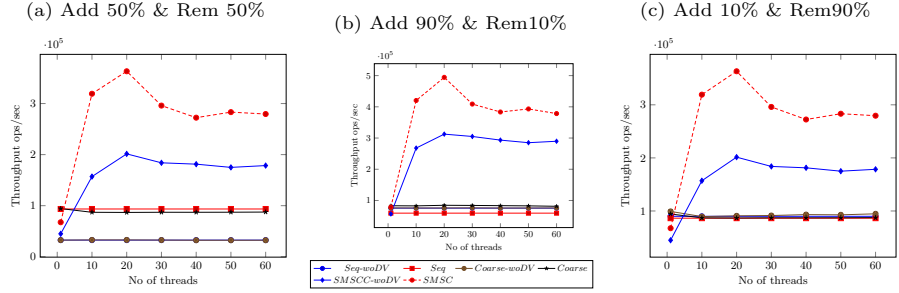
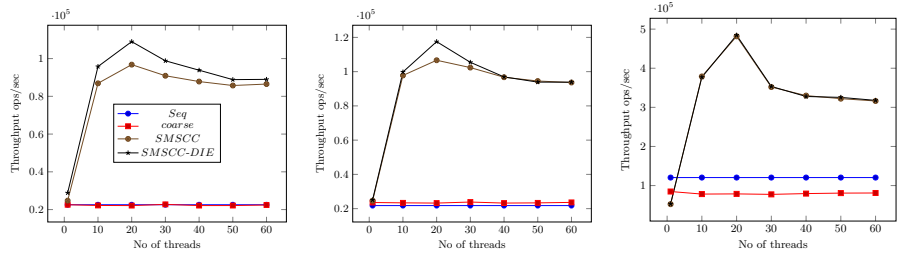
\begin{figure}
\setlength{\belowcaptionskip}{-7pt}
    \begin{subfigure}[b]{0.32\textwidth}
    \captionsetup{font=scriptsize}
    \captionsetup{font=scriptsize}
    \caption{SMISCC}
        \centering
        \resizebox{\linewidth}{!}{
	\begin{tikzpicture} [scale=0.450]
	\begin{axis}[legend style={at={(0.5,0.6)},anchor=north},
	xlabel=No of threads,
	ylabel=Throughput ops/sec, ylabel near ticks]
	\addlegendentry{$Seq$}
	\addplot table [x=Threads, y=$Seq$]{results/sccinc.dat};
	\addlegendentry{$coarse$}
	\addplot table [x=Threads, y=$coarse$]{results/sccinc.dat};
	\addlegendentry{$SMSCC$}
	\addplot table [x=Threads, y=$SMSCC$]{results/sccinc.dat};
	\addlegendentry{$SMSCC\text{-}DIE$}
	\addplot table [x=Threads, y=$SMSCC-DIE$]{results/sccinc.dat};
	\end{axis}
	\end{tikzpicture}
        }
        \label{fig:scc-inc}
    \end{subfigure}
    \begin{subfigure}[b]{0.32\textwidth}
    \captionsetup{font=scriptsize}
    \setlength{\belowcaptionskip}{-7pt}
    \caption{SMDSCC}   
    \centering
        \resizebox{\linewidth}{!}{
   	\begin{tikzpicture}[scale=0.450]
\begin{axis}[legend style={legend columns=4,{at={(1.5,-0.2)}}},
	xlabel=No of threads,
	ylabel=Throughput ops/sec, ylabel near ticks]
\addplot table [x=Threads, y=$Seq$]{results/sccdec.dat};
	\addplot table [x=Threads, y=$coarse$]{results/sccdec.dat};
	\addplot table [x=Threads, y=$SMSCC$]{results/sccdec.dat};
	\addplot table [x=Threads, y=$SMSCC-DIE$]{results/sccdec.dat};
	\end{axis}
	\end{tikzpicture}
        }
        \label{fig:scc-dec}
    \end{subfigure}
    \begin{subfigure}[b]{0.32\textwidth}
    \captionsetup{font=scriptsize}
    \setlength{\belowcaptionskip}{-7pt}
    \caption{Community Detection}
        \centering
        \resizebox{\linewidth}{!}{
           \begin{tikzpicture}[scale=0.450]
		\begin{axis}[legend style={legend columns=4,{at={(-0.7,-0.5)}}},
	xlabel=No of threads,
	ylabel=Throughput ops/sec, ylabel near ticks]
\addplot table [x=Threads, y=$Seq$]{results/scccomm.dat};
	\addplot table [x=Threads, y=$coarse$]{results/scccomm.dat};
	\addplot table [x=Threads, y=$SMSCC$]{results/scccomm.dat};
	\addplot table [x=Threads, y=$SMSCC-DIE$]{results/scccomm.dat};
	\end{axis}
	\end{tikzpicture}
        }
        \label{fig:scc-com}
    \end{subfigure}
    \vspace{-5mm}
    \setlength{\belowcaptionskip}{-13pt}
    \captionsetup{font=scriptsize}
\caption{(a)Incremental SCC(100\% Add(V+E)), (b)Decremental SCC(100\% Rem(V+E) and (c). Community detection( checking 80\% + update 20\% } 
\label{fig:scc-app}
\end{figure}

\section{Conclusion \& Future Direction}
\label{sec:conclusion-future-directions}

In this paper, we proposed a fully dynamic algorithm(SMSCC) for maintaining strongly connected component of a directed graph in a shared memory architecture. The edges/vertices are added or deleted concurrently by fixed number of threads. To the best of our knowledge, this is the first work to propose using linearizable concurrent \ds. We have constructed SCC-Graph using three type of nodes, \texttt{SCC node}, \texttt{vertex \& edge node}, which were build using list-based set, first one is unordered and later two are ordered list. We provide an empirical comparison against sequential, coarse-grained,  with different workload distributions. Also we compare the result with delete \& without delete incoming edges. We concluded that the performance show in the Fig. \ref{fig:scc} the throughput is increased between $3$ to $6$x.  Also  In Fig \ref{fig:scc-app} depicts one application of SCC-Graph to identify community in a random graph. We believe that there are huge applications in the on-line graph.

Currently the proposed update algorithms are blocking and deadlock-free. In the future, we plan to explore non-blocking(lock-free \& wait-free) variant of all the \mth{s} of SCC-Graph. We believe that one can develop a better optimization techniques to handle the SCC restoring after the edges/vertices are added or deleted.  Also we plan for other real world social graph applications.

\bibliographystyle{plain}

\bibliography{biblio}

\begin{thebibliography}{10}

\bibitem{Bader2009}
David~A. Bader, Jonathan~W. Berry, Daniel Chavarŕıa-Miranda, Kamesh Madduri,
  and Steven~C. Poulos.
\newblock Stinger: Spatio-temporal interaction networks and graphs (sting)
  extensible representation.
\newblock 2009.

\bibitem{Michael}
Michael~A. Bender, Jeremy~T. Fineman, Seth Gilbert, and Robert~E. Tarjan.
\newblock {A New Approach to Incremental Cycle Detection and Related Problems}.
\newblock {\em {ACM} Trans. Algorithms}, 12(2):14, 2016.
\newblock URL: \url{http://doi.acm.org/10.1145/2756553}, \href
  {http://dx.doi.org/10.1145/2756553} {\path{doi:10.1145/2756553}}.

\bibitem{Bloemen:2016}
Vincent Bloemen, Alfons Laarman, and Jaco van~de Pol.
\newblock Multi-core on-the-fly scc decomposition.
\newblock {\em SIGPLAN Not.}, 51(8):8:1--8:12, February 2016.
\newblock URL: \url{http://doi.acm.org/10.1145/3016078.2851161}, \href
  {http://dx.doi.org/10.1145/3016078.2851161}
  {\path{doi:10.1145/3016078.2851161}}.

\bibitem{Cormen}
Thomas~H. Cormen, Charles~E. Leiserson, Ronald~L. Rivest, and Clifford Stein.
\newblock {\em Introduction to Algorithms {(3.} ed.)}.
\newblock {MIT} Press, 2009.
\newblock URL: \url{http://mitpress.mit.edu/books/introduction-algorithms}.

\bibitem{Demetrescu2010}
Camil Demetrescu, David Eppstein, Zvi Galil, and Giuseppe~F. Italiano.
\newblock In Mikhail~J. Atallah and Marina Blanton, editors, {\em Algorithms
  and Theory of Computation Handbook}, chapter Dynamic Graph Algorithms, pages
  9.1--9.28. Chapman \& Hall/CRC, 2010.

\bibitem{Green2016}
Oded Green and David~A. Bader.
\newblock custinger: Supporting dynamic graph algorithms for gpus.
\newblock {\em 2016 IEEE High Performance Extreme Computing Conference (HPEC)},
  pages 1--6, 2016.

\bibitem{Henzinger:1999}
Monika~R. Henzinger and Valerie King.
\newblock Randomized fully dynamic graph algorithms with polylogarithmic time
  per operation.
\newblock {\em J. ACM}, 46(4):502--516, July 1999.
\newblock URL: \url{http://doi.acm.org/10.1145/320211.320215}, \href
  {http://dx.doi.org/10.1145/320211.320215} {\path{doi:10.1145/320211.320215}}.

\bibitem{opodis_Herlihy}
Maurice Herlihy and Nir Shavit.
\newblock {On the Nature of Progress}.
\newblock In {\em Principles of Distributed Systems - 15th International
  Conference, {OPODIS} 2011, Toulouse, France, December 13-16, 2011.
  Proceedings}, pages 313--328, 2011.
\newblock URL: \url{http://dx.doi.org/10.1007/978-3-642-25873-2_22}, \href
  {http://dx.doi.org/10.1007/978-3-642-25873-2_22}
  {\path{doi:10.1007/978-3-642-25873-2_22}}.

\bibitem{HerlWing:1990:TPLS}
Maurice~P. Herlihy and Jeannette~M. Wing.
\newblock {Linearizability: a correctness condition for concurrent objects}.
\newblock {\em ACM Trans. Program. Lang. Syst.}, 12(3):463--492, 1990.
\newblock \href {http://dx.doi.org/http://doi.acm.org/10.1145/78969.78972}
  {\path{doi:http://doi.acm.org/10.1145/78969.78972}}.

\bibitem{Hirschberg:1979}
D.~S. Hirschberg, A.~K. Chandra, and D.~V. Sarwate.
\newblock Computing connected components on parallel computers.
\newblock {\em Commun. ACM}, 22(8):461--464, August 1979.
\newblock URL: \url{http://doi.acm.org/10.1145/359138.359141}, \href
  {http://dx.doi.org/10.1145/359138.359141} {\path{doi:10.1145/359138.359141}}.

\bibitem{HopcroftTarjan:1973}
John Hopcroft and Robert Tarjan.
\newblock Algorithm 447: Efficient algorithms for graph manipulation.
\newblock {\em Commun. ACM}, 16(6):372--378, June 1973.
\newblock URL: \url{http://doi.acm.org/10.1145/362248.362272}, \href
  {http://dx.doi.org/10.1145/362248.362272} {\path{doi:10.1145/362248.362272}}.

\bibitem{Michael:2004}
Maged~M. Michael.
\newblock Hazard pointers: Safe memory reclamation for lock-free objects.
\newblock {\em IEEE Trans. Parallel Distrib. Syst.}, 15(6):491--504, June 2004.
\newblock URL: \url{http://dx.doi.org/10.1109/TPDS.2004.8}, \href
  {http://dx.doi.org/10.1109/TPDS.2004.8} {\path{doi:10.1109/TPDS.2004.8}}.

\bibitem{PeriSS16}
Sathya Peri, Muktikanta Sa, and Nandini Singhal.
\newblock {Maintaining Acyclicity of Concurrent Graphs}.
\newblock {\em CoRR}, abs/1611.03947, 2016.
\newblock URL: \url{http://arxiv.org/abs/1611.03947}.

\bibitem{Shiloach:1981}
Yossi Shiloach and Shimon Even.
\newblock An on-line edge-deletion problem.
\newblock {\em J. ACM}, 28(1):1--4, January 1981.
\newblock URL: \url{http://doi.acm.org/10.1145/322234.322235}, \href
  {http://dx.doi.org/10.1145/322234.322235} {\path{doi:10.1145/322234.322235}}.

\bibitem{ShiloachVishkin:1982}
Yossi Shiloach and Uzi Vishkin.
\newblock An o(logn) parallel connectivity algorithm.
\newblock {\em Journal of Algorithms}, 3(1):57 -- 67, 1982.
\newblock URL:
  \url{http://www.sciencedirect.com/science/article/pii/0196677482900086},
  \href {http://dx.doi.org/https://doi.org/10.1016/0196-6774(82)90008-6}
  {\path{doi:https://doi.org/10.1016/0196-6774(82)90008-6}}.

\bibitem{TimnatBKP12}
Shahar Timnat, Anastasia Braginsky, Alex Kogan, and Erez Petrank.
\newblock Wait-free linked-lists.
\newblock In {\em Principles of Distributed Systems, 16th International
  Conference, {OPODIS} 2012, Rome, Italy, December 18-20, 2012. Proceedings},
  pages 330--344, 2012.
\newblock URL: \url{https://doi.org/10.1007/978-3-642-35476-2_23}, \href
  {http://dx.doi.org/10.1007/978-3-642-35476-2_23}
  {\path{doi:10.1007/978-3-642-35476-2_23}}.

\bibitem{Joseph12}
A~Gitter Z~Bar~Joseph and I~Simon.
\newblock Studying and modelling dynamic biological processes using time-series
  gene expression data.
\newblock In {\em Nature Publishing Group, a division of Macmillan Publishers
  Limited. All Rights Reserved.}, pages 552--564, 2012.
\newblock URL: \url{https://www.ncbi.nlm.nih.gov/pubmed/22805708}.

\end{thebibliography}




\end{document}